\documentclass[journal]{IEEEtran}

\usepackage{amsmath}
\usepackage[ruled,linesnumbered]{algorithm2e}
\usepackage{amsthm}

\usepackage{booktabs}
\newcommand{\ra}[1]{\renewcommand{\arraystretch}{#1}}
\usepackage{amssymb}  % assumes amsmath package installed
\usepackage{mathtools}
\usepackage{commath}
\usepackage{subcaption}

\usepackage{enumitem}   

\usepackage{xargs}
\usepackage{todonotes}
\newcommandx{\declan}[2][1=]{\todo[linecolor=blue,backgroundcolor=blue!25,bordercolor=blue,size=small,author=Declan,#1]{#2}}
\newcommandx{\iman}[2][1=]{\todo[linecolor=green,backgroundcolor=green!25,bordercolor=green,size=small,author=Iman,#1]{#2}}
\newcommandx{\airlie}[2][1=]{\todo[linecolor=purple,backgroundcolor=purple!25,bordercolor=purple,size=small,author=Airlie,#1]{#2}}

\newcounter{prob}
\newtheorem{prob}[prob]{\bf Problem}

\newtheorem{lemm}{\bf Lemma}[]

\newtheorem{prop}[lemm]{\bf Proposition}

\allowdisplaybreaks % to allow breaking align blocks and others over two columns/pages

\begin{document}
\title{Fast Spline Trajectory Planning: \\ Minimum Snap and Beyond}

\author{Declan Burke,~\IEEEmembership{Student Member,~IEEE,}
        Airlie Chapman,~\IEEEmembership{Member,~IEEE},
        Iman Shames,~\IEEEmembership{Member,~IEEE}% <-this % stops a space
\thanks{The first and second authors are with the Department of Mechanical Engineering, The University of Melbourne, VIC, 3010, Australia and the third with the Research School of Engineering, The Australian National University, ACT, 0200,  Australia. Emails: {\tt\small \{declanb@student., airlie.chapman@\}unimelb.edu.au, iman.shames@anu.edu.au}}}

\markboth{IEEE Transactions on Robotics}%
{}

\maketitle

\begin{abstract}
In this paper, we study spline trajectory generation via the solution of two optimisation problems: (i) a quadratic program (QP) with linear equality constraints and (ii) a nonlinear and nonconvex optimisation program. We propose an efficient algorithm to solve (i), which we then leverage to use in an iterative algorithm to solve (ii). Both the first algorithm and each iteration of the second algorithm have linear computational complexity in the number of spline segments. The scaling of each algorithm is such that we are able to solve the two problems faster than state-of-the-art methods and in times amenable to real-time trajectory generation requirements. The trajectories we generate are applicable to differentially flat systems, a broad class of mechanical systems, which we demonstrate by planning trajectories for a quadrotor. 
\end{abstract}


\IEEEpeerreviewmaketitle
% \iman[inline]{Only use  present tenses (simple and perfect). Don't use past tenses. You can use the future tense in some rare occasions. E.g., outlining future directions.}
\section{Introduction}

In trajectory planning for quadrotors, minimum snap, i.e., fourth derivative, splines are commonly employed~\cite{mahony2012multirotor}. Minimum snap trajectory pioneers Mellinger and Kumar solve an equality constrained quadratic program (QP) to generate spline trajectories~\cite{Mellinger2011}. The numerical stability of this problem is considered by Richter {\it et al.}~\cite{Richter2016} and further by de Almeida and Akella~\cite{DeAlmeida2017}, who each propose methods with better numerical stability. In a previous work, we propose a well-conditioned formulation of Mellinger and Kumar's problem as well as an algorithm that solves it with linear computational complexity in the number of segments in the spline trajectory~\cite{burke2020}. Subsequent to our work, an upcoming paper by Wang {\it et al.} improves upon our results by presenting an algorithm without matrix inverse calculations and thus with a lower computational complexity coefficient \cite{wang2020generating}. Another means of generating spline trajectories is achieved by formulating a problem in which the time between spline segments is included as a variable over which to optimise~\cite{Mellinger2011},~\cite{Richter2016}. The trajectories generated by this problem have smaller snap than those generated by the aforementioned QP, but this problem is nonlinear and nonconvex~\cite{burri2015real}. In solving it, de Almeida {\it et al.} encounter long computation times and note that it is ill-suited to real-time applications~\cite{de2019real}. Also as part of their upcoming work, Wang {\it et al.} propose computationally efficient expressions for a gradient descent method used to solve a closely related problem~\cite{wang2020generating}.

The recent developments in minimum snap trajectory planning are also of benefit for trajectory planning for systems beyond quadrotors. Polynomials may be used for quadrotor trajectory planning because quadrotors are differentially flat systems, that is, all states and inputs of the system can be calculated from a set of outputs without integration~\cite{Van_Nieuwstadt1998-vs}. The class of differentially flat systems is extensive and incorporates many mechanical systems, including robot arms and cars with trailers~\cite{Murray95differentialflatness}. Motivated by the wealth of applications that could benefit from efficient spline generation algorithms, we leverage the recent developments in minimum snap trajectory planning for general spline trajectory planning. Through the use of efficient methods, we aim to enable real-time trajectory generation for platforms with varying requirements, such as the number of segments in their spline trajectories or the dimension of their trajectories.

\begin{figure}[t]
\centering
{\includegraphics[width=0.98\linewidth]{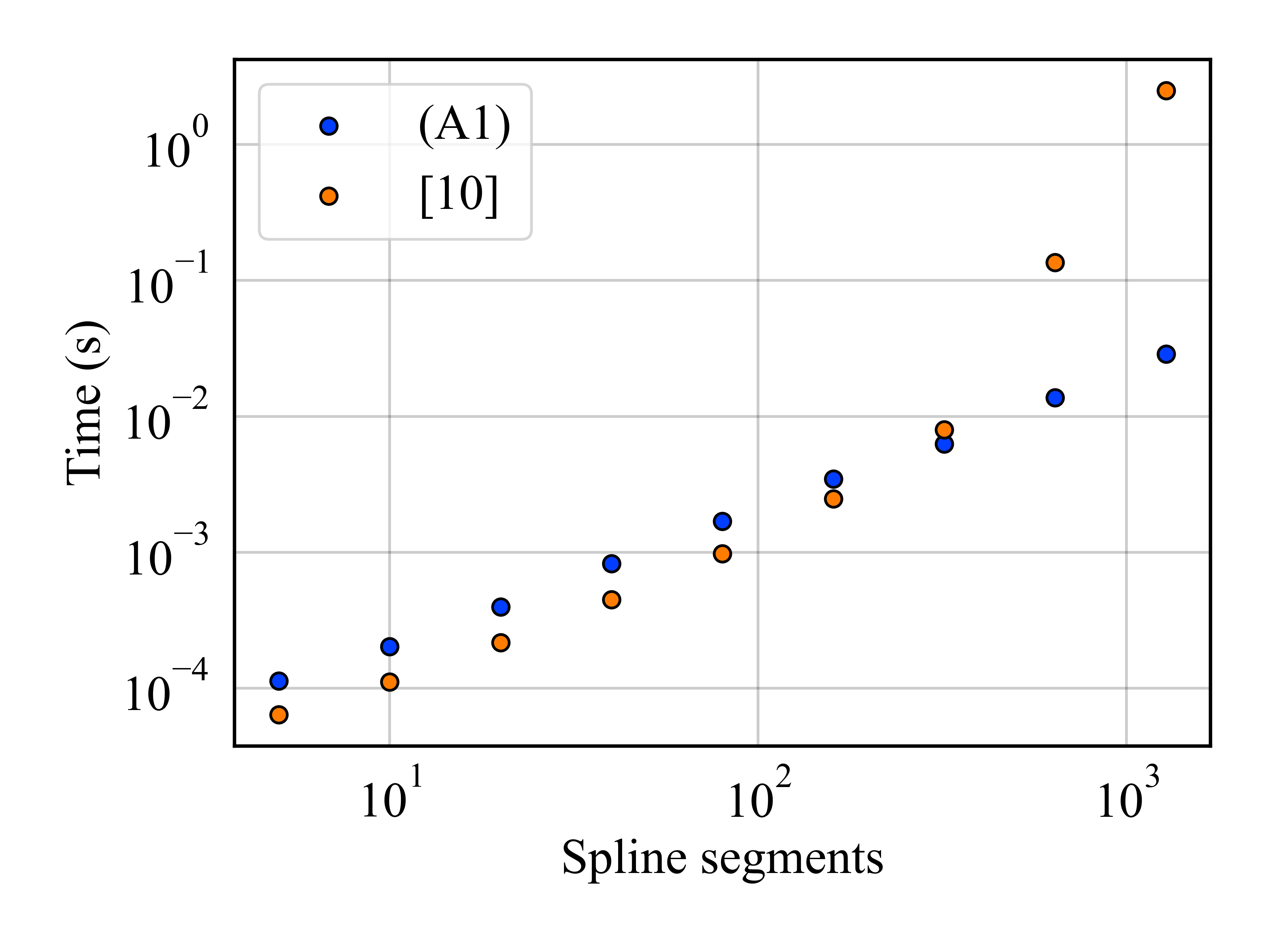}}
\caption{The computational time for generating minimum snap spline trajectories of $l$ segments using Algorithm~\ref{algo1} (A1) (blue) and a state-of-the-art solver~\cite{eth_implementation} (orange). The slope of the logarithmic data from Algorithm~\ref{algo1} is approximately one corresponding to linear computational complexity in $l$, whereas the slope of the benchmark corresponds to approximately quadratic computational complexity in $l$.}
\end{figure}

In this paper, we study and solve two problems in order to generate spline trajectories. The first problem, which we coin the {\it fixed-time problem}, is a QP with linear equality constraints. We develop a previous result~\cite{burke2020} and propose an algorithm to solve the fixed-time problem with linear computational complexity in the number of segments in the spline trajectory. Further, we present an equivalent reformulation of the underlying optimisation program that is better conditioned. The second problem, which we name the {\it variable-time problem}, is a nonlinear and nonconvex minimisation problem. To solve the variable-time problem we employ a gradient-descent approach. We propose a set of efficient expressions for calculating the gradient of the program's objective function thereby reducing the method's computational complexity. We evidence our methods by generating minimum snap trajectories for a quadrotor in real-time. We also demonstrate how our methods can be included in a broader suite of path planning algorithms by implementing an RRT* algorithm that outputs a tree of minimum snap trajectories.

In summary, the two methods we propose are efficient and scalable, thereby we are able to solve the two problems faster than the state-of-the-art.

The paper is structured as follows. In the next section, we introduce requisite notation and formulate an optimisation program to generate spline trajectories. In Section~\ref{sec::solution}, we detail an algorithm for solving the optimisation program. In Section~\ref{sec::time_optimising}, we present another optimisation program for generating spline trajectories and propose and analyse a method for its solution. In Section~\ref{sec::experiments}, we demonstrate trajectories on two experimental platforms. Concluding remarks come in the end.

\section{Fixed-Time Problem Formulation}\label{sec::formulation}

\subsection{Notation}

Scalars are written with lowercase Greek letters, vectors in lowercase Latin letters, matrices in uppercase Latin letters, and sets in calligraphic font. The only exception to this is using $h$, $i$, $j$, $k$, $l$, $m$, $n$, $p$ and $q$ to denote non-negative integers and using $J$ to denote the scalar cost function. Sequences of scalars, vectors or matrices are enumerated using subscripts $x_1,~x_2,~x_3,\ldots$. Unless specified otherwise, elements of a vector or matrix are denoted using parentheses indexed with subscripts, e.g., $(x_i)_j$ is the $j$th element of the vector $x_i$ and $(A_i)_{j,k}$ is the $(j,k)$th element of the matrix $A_i$. Euclidean space of dimension $n$ is written as $\mathbb{R}^n$. The $q$th time derivative of $x$ is denoted with $x^{(q)}\coloneqq\frac{\mathrm{d}^qx}{\mathrm{d}t^q}$. The cardinality of a set $\mathcal{X}$ is written as $|\mathcal{X}|$. The identity matrix of dimension $k$ is denoted with $I_k$. For the sequences $\eta_1,~\eta_2,\ldots$ and $\nu_1,~\nu_2,\ldots$, $\eta_i = O(\nu_i)$ is written if there exists a positive real constant $\alpha$ such that $|\eta_i|\leq \alpha|\nu_i|$ for $i$ sufficiently large. Let ${\bf P}_+:\mathbb{R}^n\rightarrow \mathbb{R}^n$ project onto the positive orthant such that $\big({\bf P}_+(x)\big)_i=\max\{x_i,0\}$ for $i=1,\ldots,n$. Given sets of column vectors $\{x_1,\dots,x_k\}$ and matrices $\{A_1,\dots,A_k\}$, let
\begin{align*}
    \mathrm{vec}(\{x_1,\dots,x_k\})&\coloneqq\begin{bmatrix}
    x_1 \\ \vdots \\ x_k
    \end{bmatrix},\\\mathrm{diag}(\{A_1,\dots,A_k\})&\coloneqq\begin{bmatrix}
    A_1 &  &  \\ & \ddots & \\ & & A_k
    \end{bmatrix}.
\end{align*}

\subsection{Optimization formulation of spline interpolation}

We study the continuous spline $\sigma(\tau):\mathbb{R}\rightarrow\mathbb{R}$ that is parametrised by time $\tau\in\mathbb{R}$
\begin{align*}
    \sigma(\tau)&=\begin{cases}
    \sigma_1(\tau), & \tau_0\leq \tau\leq \tau_1, \\
    \quad \vdots & \\
    \sigma_l(\tau), & \tau_{l-1}\leq \tau\leq \tau_l,
    \end{cases}
\end{align*}
with segments $\sigma_i(\tau)$, $i=1,\ldots,l$, defined with respect to the times $t=[\tau_0,\ldots,\tau_l]^T\in\mathbb{R}^{l+1}$. The segments are polynomials of order $n-1$ represented as
\begin{align*}
    \sigma_i(\tau)&=a_i^T z(\tau)
\end{align*}
where $z(\tau):\mathbb{R}\rightarrow\mathbb{R}^n$ is a vector of the monomial basis functions and $a_i\in\mathbb{R}^{n}$ is a vector of coefficients for $i=1,\ldots,l$. We stack the vectors of coefficients as $a=\mathrm{vec}(\{a_1,\ldots,a_l\})\in\mathbb{R}^{ln}$.
    
As in the problem of Hermite interpolation~\cite{Davis1963}, we generate a spline such that the spline and its derivatives take the values
\begin{gather*}
    (f_i)_q=\sigma_i^{(q-1)}(\tau_{i-1}), \\
    (f_i)_{q+k}=\sigma_i^{(q-1)}(\tau_i),
\end{gather*}
for $q=1,\ldots,k$ for integer $k$ such that $f_i\in\mathbb{R}^{2k}$ for $i=1,\ldots,l$. We stack the vectors of derivatives as $f=\mathrm{vec}(\{f_1,\ldots,f_l\})\in\mathbb{R}^{2kl}$.% We introduce the spline and its derivatives $s(\tau)=[\sigma(\tau),\sigma^{(1)}(\tau),\ldots,\sigma^{(k-1)}(\tau)]^T:\mathbb{R}\rightarrow\mathbb{R}^k$ and $f=\mathrm{vec}(\{f_1,\ldots,f_l\})\in\mathbb{R}^{2kl}$.

We generate the spline by finding the minimising argument of the optimisation program
% \declan[inline]{Will change when paper is finalised.}
% \airlie[inline]{optimization is with a 'z' - check throughout}
% \iman[inline]{So American. Using `s' is fine. But if you want to go with z, be consistent. so all spelling choices should be switched to the American version, e.g. neighbour to neighbor, etc.}
\begin{subequations} \label{eq::form1}
    \begin{gather}
        J(t)= \min_{\sigma(\tau), f} \int_{\tau_0}^{\tau_l} \big(\sigma^{(k-1)}(\tau)\big)^2d\tau, \label{eq::form1::cost}\\
        \sigma^{(q-1)}_i(\tau_{i-1})=(f_i)_q, \quad i=1,\ldots,l,~q=1,\ldots,k, \label{eq::form1::constr1} \\
        \sigma^{(q-1)}_i(\tau_{i})=(f_i)_{k+q}, \quad i=1,\ldots,l,~q=1,\ldots,k, \label{eq::form1::constr2} \\
        (f_i)_{q+k}=(f_{i+1})_q, \quad i=1,\ldots,l-1,~q=1,\ldots,k, \label{eq::form1::continuity} \\
        (f_i)_p = \phi,\quad  \forall (i,p,\phi)\in \mathcal{P}, \label{eq::form1::set1}
        %\forall(i,q, \phi^i_q)\in \mathcal{P}, \label{eq::form1::set1}
        % f^{(q)}_{i+1,0} = \phi^{(q)}_{i},\quad  \forall(i,q)\in \mathcal{P}, \label{eq::form1::set2}
    \end{gather}
\end{subequations}
% \begin{subequations} \label{eq::form1}
%     \begin{gather}
%         J(t)= \min_{\sigma(\tau), f} \int_{\tau_0}^{\tau_l} \big(w^T s(\tau)\big)^2d\tau, \label{eq::form1::cost}\\
%         \sigma^{(q-1)}_i(\tau_{i-1})=(f_i)_q, \quad i=1,\ldots,l,~q=1,\ldots,k, \label{eq::form1::constr1} \\
%         \sigma^{(q-1)}_i(\tau_{i})=(f_i)_{k+q}, \quad i=1,\ldots,l,~q=1,\ldots,k, \label{eq::form1::constr2} \\
%         (f_i)_{q+k}=(f_{i+1})_q, \quad i=1,\ldots,l-1,~q=1,\ldots,k, \label{eq::form1::continuity} \\
%         (f_i)_p = \phi,\quad  \forall (i,p,\phi)\in \mathcal{P}, \label{eq::form1::set1}
%         %\forall(i,q, \phi^i_q)\in \mathcal{P}, \label{eq::form1::set1}
%         % f^{(q)}_{i+1,0} = \phi^{(q)}_{i},\quad  \forall(i,q)\in \mathcal{P}, \label{eq::form1::set2}
%     \end{gather}
% \end{subequations}
where % $w\in\mathbb{R}^k$ and 
the set $\mathcal{P}$ is a collection of triplets $(i,p,\phi)$ with integer $i\in\{1,\ldots,l\}$ and $p\in\{1,\ldots 2k\}$ and $\phi\in\mathbb{R}$.

The parameters of the program are the vector $t$, the integer $k$ and the set $\mathcal{P}$. We explicitly represent the objective function in terms of the parameter $t$ as it is a focus of our study later in the paper and we note that $k$ and $\mathcal{P}$ parametrise \eqref{eq::form1} implicitly. The formulator of the problem chooses $t$ to temporally constrain the spline. The choice of $k$ determines the total derivative to be minimised as the cost \eqref{eq::form1::cost} and the degree of continuity enforced through \eqref{eq::form1::continuity}. Each $(i,p,\phi)\in\mathcal{P}$ restricts elements of feasible $f$ through the interpolation constraints \eqref{eq::form1::set1}. The size of $\mathcal{P}$ is established as follows. For each $i$, let $\mu_i^-=|\{p|p\in\{1,\ldots,k\},(i,p,\phi)\in\mathcal{P}\}|$ and $\mu_i^+=|\{p|p\in\{k+1,\ldots,2k\},(i,p,\phi)\in\mathcal{P}\}|$. Further let $m=\sum_{i=1}^{l} \mu_i^- + \mu_i^+$. Hence, there are $m=|\mathcal{P}|$ constraints constituted by \eqref{eq::form1::set1}. %If not constrained via $\mathcal{P}$, the elements of $f$ are decision variables whose values are determined by the optimality conditions of the cost \eqref{eq::form1::cost}. % This cost, and hence optimality conditions, are determined by $w$. 

\subsection{A compact representation}

In this subsection, we express \eqref{eq::form1} more compactly as part of the problem we call the fixed-time problem. We will first present the problem  and then detail the vectors and matrices employed in the formulation of the associated optimisation program. 

\begin{prob}[Fixed-Time Problem] \label{prob1} For a given $t$, let $a^\star$ and $f^\star$ be the solutions of
\begin{subequations} \begin{align}
        J(t)= \min_{a, f} \quad & a^TH(t)a, \label{eq::form2::cost} \\
        \text{s.t.} \quad & V(t) a = f, \label{eq::form2::vand}\\
        &E f = 0, \\
        &P f = b. \label{eq::form2::phi}
    \end{align}\label{eq::form2}
\end{subequations}
Find the continuous spline $\sigma(\tau)$ with coefficients $a^\star$.
\end{prob}
% \begin{prob}[Fixed-Time Problem] \label{prob1} For a given $t\in\mathbb{R}^{l+1}$, let $a^\star\in\mathbb{R}^{ln}$ and $f^\star\in\mathbb{R}^{2kl}$ be the minimising arguments of the optimisation program 
% \begin{subequations} \begin{align}
%         J(t)= \min_{a, f} \quad & a^TH(t)a, \\
%         \text{s.t.} \quad & V(t) a = f, \label{eq::form2::vand}\\
%         &E f = 0, \\
%         &P f = b. \label{eq::form2::phi}
%     \end{align}\label{eq::form2}
% \end{subequations}
% Find the continuous spline $\sigma:\mathbb{R}\rightarrow\mathbb{R}$ with coefficients $a^\star$.
% \end{prob}

To compactly represent the constraints \eqref{eq::form1::constr1} and \eqref{eq::form1::constr2}, we introduce the matrices $V_i(\tau_{i-1}, \tau_i):\mathbb{R}^2\rightarrow\mathbb{R}^{2k\times n}$ for $i=1,\ldots,l$
\begin{align*}
    V_i(\tau_{i-1}, \tau_i)&=\begin{bmatrix}
    W(\tau_{i-1}) \\
    W(\tau_i)
    \end{bmatrix},
\end{align*}
where $W(\tau):\mathbb{R}\rightarrow\mathbb{R}^{k\times n}$ is a matrix parametrised by time $\tau$ such that $\big(W(\tau)\big)_{q,j}=\big(z^{(q-1)}(\tau)\big)_j$. % The matrix $V_0(\tau_{i-1},\tau_i)$ is a transposed confluent Vandermonde matrix~\cite{higham2002accuracy}. 
For brevity, $V_i(\tau_{i-1},\tau_i)$ is written as $V_i$, for $i=1,\ldots,l$. Thus, we form the block-diagonal matrix $V(t):\mathbb{R}^{l+1}\rightarrow\mathbb{R}^{2kl\times ln}$ as $V(t)=\mathrm{diag}(\{V_1,\ldots,V_l\})$.

To represent the cost \eqref{eq::form1::cost}, we introduce the matrix
\begin{align*}
    H_i (\tau_{i-1}, \tau_i) &= \int_{\tau_{i-1}}^{\tau_{i}} z^{(k-1)}(\tau)\big(z^{(k-1)}(\tau)\big)^T d\tau,
\end{align*}
% \begin{align*}
%     H_i (\tau_{i-1}, \tau_i) &= \int_{\tau_{i-1}}^{\tau_{i}}W(\tau)^Tww^TW(\tau) d\tau,
% \end{align*}
where $H_i(\tau_{i-1}, \tau_i):\mathbb{R}^2\rightarrow\mathbb{R}^{n\times n}$. We abbreviate $H_i(\tau_{i-1},\tau_i)$ as $H_i$, for $i=1,\ldots,l$, thus, the Hessian of the cost \eqref{eq::form1::cost} is $H(t):\mathbb{R}^{l+1}\rightarrow\mathbb{R}^{ln\times ln}$ such that $H(t)=\mathrm{diag}(\{H_1,\ldots,H_l\})$.

% The cost \eqref{eq::form1::cost} is a quadratic with a Hessian with block-diagonal elements, for $i=1,\ldots,l$,
% \begin{align*}
%     H_0(\tau_{i-1}, \tau_i) &= \int_{\tau_{i-1}}^{\tau_{i}}W(\tau)^Tcc^TW(\tau) d\tau,
% \end{align*}
% where $H_0(\tau_{i-1}, \tau_i)\in\mathbb{R}^{n\times n}$.

% We abbreviate $V_0(\tau_{i-1},\tau_i)$ and $H_0(\tau_{i-1},\tau_i)$ as $V_i$ and $H_i$, respectively, for $i=1,\ldots,l$. We then form the block-diagonal matrices $V(t)=\mathrm{diag}(\{V_1,\ldots,V_l\})\in\mathbb{R}^{2kl\times ln}$ and $H(t)=\mathrm{diag}(\{H_1,\ldots,H_l\})\in\mathbb{R}^{ln\times ln}$.

Constraints \eqref{eq::form1::continuity} and \eqref{eq::form1::set1} can be written compactly with constant matrices of ones and zeros. To this end, we introduce the matrix $E \in\mathbb{R}^{k(l-1) \times 2kl}$ 
\begin{align*}
    E = \begin{bmatrix}
            0 & I_{k} & -I_{k} & 0 & 0 & \hdots & 0 & 0 & 0\\
            0 & 0 & 0 & I_{k} & -I_{k} & \hdots & 0 & 0 & 0   \\
            &  & \vdots & & & \ddots & & \vdots & \\
            0 & 0 & 0 & 0 & 0 & \hdots & I_{k} & -I_{k} & 0 
        \end{bmatrix},
\end{align*}
and the matrix $P=\mathrm{diag}(\{P_1,\ldots,P_l\})\in\mathbb{R}^{m\times 2kl}$, which is block-diagonal with elements $P_i\in\mathbb{R}^{(\mu_i^- + \mu_i^+)\times 2k}$ and $(P_i)_{j,p}=1$ if $(i,p,\phi)\in \mathcal{P}$ and $j\in\{1, \ldots,\mu_i^- + \mu_i^+\}$. We note that $P$ has one nonzero element per row.

Let $b=\mathrm{vec}(\{b_1,\ldots,b_l\})\in\mathbb{R}^{m}$ with  $b_i\in\mathbb{R}^{\mu_i^- + \mu_i^+}$, $i=1,\ldots,l$, such that $(b_i)_j=\phi$ for each $(i,p,\phi)\in\mathcal{P}$ with $j\in\{1,\ldots,\mu_i^-+\mu_i^+\}$ and is zero otherwise.

\section{Solving the Fixed-Time Problem}\label{sec::solution}
In this section, we present an algorithm of linear computational complexity in the number of spline segments $l$ for solving the fixed-time problem. We also present an equivalent reformulation of \eqref{eq::form2} that may be solved in its place to avoid numerical errors due to ill-conditioning.

\subsection{Solution Algorithm}
  \begin{algorithm}[t]
    \SetAlgoLined
    \For{i=1,\ldots, l} {
        construct $V_i,H_i,Z_i,b_{i}$; \\
        partition $A_i,B_i,C_i,c_i^-, c_i^+$; \\
        $\overline{A}_i {\gets\begin{cases}
        A_1 & i=1, \\
        A_{i}+C_{i-1}-B_{i-1}^T\overline{A}_{i-1}^{-1}B_{i-1} & i=2,\ldots,l;
        \end{cases}}$ \\
        $\overline{c}_i^- {\gets \begin{cases}
        g_1^- & i=1, \\
        g_i^- + g_{i-1}^- -B_{i-1}^T\overline{A}_{i-1}^{-1}\overline{c}_{i-1}^-
        % \begin{aligned}
        % & g_i^- + g_{i-1}^- \\
        % & \quad -B_{i-1}^TA_{i-1}^{-1}\overline{c}_{i-1}^-
        % \end{aligned}
        & i=2,\ldots,l;
        \end{cases}}$
    }
    \For{i=l,\ldots, 1} {
        $g_i^-\gets  \overline{A}^{-1}_i(\overline{c}_i^--B_ig_i^+) $; \\
        $g_i^+ {\gets \begin{cases}
        \begin{aligned}
        & (C_l-B_l^T\overline{A}_l^{-1}B_l)^{-1}  \\
        & \quad (c_l^+-B_l^T\overline{A}_l^{-1}\overline{c}_l^-)
        \end{aligned} 
        & i=l, \\
        g_{i+1}^- & i=l-1,\ldots,1;
        \end{cases}}$ \\
        solve for $f_i$ \eqref{eq::proj1}; \\
        % Can use divided differences to solve for $p_i$ here.
        solve for $a_i$ using \eqref{eq::form2::vand};
    } 
 \caption{Algorithm of Linear Computational Complexity for Solving the Fixed-Time Problem} \label{algo1}
 \end{algorithm}
% \subsection{The Algorithm}
To present our algorithm for solving \eqref{eq::form1}, first, we introduce the variables used in Algorithm~\ref{algo1} and then present expressions in terms of their block diagonal components. Let
\begin{align}\label{eq::proj1}
    f &= \overline{f} + Z g,
\end{align}
where $\overline{f} \in \mathbb{R}^{2kl}$ and $Z\in \mathbb{R}^{2kl\times (2kl-m)}$ are such that $P \overline{f}=b$ and $PZ=0$, and $g\in \mathbb{R}^{2kl-m}$.  Under the change of variables \eqref{eq::proj1}, the constraint \eqref{eq::form2::phi} is satisfied for all $g$. The vector $g$ can be thought of as the free derivative values. The notation $\overline{f}$ is used stylistically to reflect that $\overline{f}$ is not a decision variable and its components are either zeros or the parameters $\phi$ of the $(i,p,\phi)\in\mathcal{P}$.

We partition $\overline{f},~g$ and $Z$ into vectors and matrices of size  $\overline{f}_{i}\in \mathbb{R}^{2k}$, $g_{i}\in \mathbb{R}^{2k-\mu_i^- -\mu_i^+}$ and $Z_{i}\in \mathbb{R}^{2k \times (2k -\mu_i^- -\mu_i^+)}$, for $i=1,\ldots,l$, such that $\overline{f}=\mathrm{vec}(\{\overline{f}_1,\ldots,\overline{f}_l\})$, $g=\mathrm{vec}(\{g_1,\ldots,g_l\})$ and  $Z=\mathrm{diag}(\{Z_{1}, \ldots, Z_{l}\})$. Further, let $g_i=\mathrm{vec}(\{g_i^-, g_i^+\})$ where  $g_i^-\in\mathbb{R}^{k-\mu_i^-}$ and $g_i^+\in\mathbb{R}^{k-\mu_i^+}$. For $i=1,\ldots,l$, let
\begin{subequations}\label{eq::partition}
    \begin{align}
    Z_i^TV_i^{-T}H_iV_i^{-1}Z_i &= \begin{bmatrix}
    A_i & B_i \\
    B_i^T & C_i
    \end{bmatrix}, \\
    Z_i^TV_i^{-T}H_iV_i^{-1}
    \overline{f}_{i}&= \begin{bmatrix}
    c_i^- \\
    c_i^+
    \end{bmatrix}
\end{align}
\end{subequations}
where $A_i\in\mathbb{R}^{(k-\mu_i^-)\times (k-\mu_i^-)}$, $B_i\in\mathbb{R}^{(k-\mu_i^-)\times (k-\mu_i^+)}$, $C_i\in\mathbb{R}^{(k-\mu_i^+)\times (k-\mu_i^+)}$, $c_i^-\in\mathbb{R}^{k-\mu_i^-}$ and $c_i^+\in\mathbb{R}^{k-\mu_i^+}$. 

\begin{prop}\label{prop::alg_complexity}
Algorithm~\ref{algo1} solves \eqref{eq::form1} in $O(k^3l)$. 
\end{prop}
\begin{proof}
    See Appendix~\ref{app::alg_complexity}.
\end{proof}

We also note that an algorithm with low computational complexity in $l$ is important given the other variable in the complexity bound $k$ is practically limited by application. For example, $k=4$ is used to minimise the jerk of robot arm trajectories~\cite{Simon1993} and $k=5$ is used to minimise the snap of quadrotor trajectories~\cite{Mellinger2011}. Linear in $l$, our algorithm's computational complexity facilitates spline trajectories of many segments.

\subsection{A better conditioned formulation}

% \declan[inline]{A well received part of the IROS paper was the comments on conditioning. We may need to include a further investigation. }

It is known that \eqref{eq::form2} has matrices that are poorly conditioned and that significant numerical errors are encountered when generating trajectories of more than $50$ segments~\cite{DeAlmeida2017}. To remedy the numerical instability of \eqref{eq::form2}, we employ the same better conditioned matrices as~\cite{burke2020}. The reformulation is
% \todo[inline]{One does not achieve a reformuation. One arrives at one or just simplly write what really matters.}
\begin{subequations} \label{eq::form_nondim}
    \begin{align}
        J(t)=\min_{a,f} \quad &a^T \widetilde{H}(t) a, \label{eq::form_nondim::cost}\\
        \text{s.t.} \quad & G(t)\widetilde{V}a =f, \label{eq::form_nondim::vand} \\
        & E f =0, \label{eq::form_nondim::const1} \\
        & P f = b, \label{eq::form_nondim::const2}
    \end{align}
\end{subequations}
where $\widetilde{V}=\mathrm{diag}\{V_1(-1,1),\ldots,V_l(-1,1)\}\in\mathbb{R}^{2kl\times ln}$ and  $\widetilde{H}(t)=\mathrm{diag}\{(2/(\tau_1-\tau_0))^{2k-1}H_1(-1,1),\ldots,(2/(\tau_l-\tau_{l-1}))^{2k-1}H_l(-1,1)\}\in\mathbb{R}^{ln\times ln}$. % We introduce the vector of the difference between times $u=[\delta_1,\ldots,\delta_l]^T\in\mathbb{R}^l$ where $\delta_i=\tau_i-\tau_{i-1}$. 
The matrix $G(t)$ has block-diagonal submatrices $G_i(\tau_{i-1},\tau_i): \mathbb{R}^2\rightarrow \mathbb{R}^{2k\times 2k}$ such that $G_i(\tau_{i-1},\tau_i)=\mathrm{diag}\{1, 2/(\tau_i-\tau_{i-1}), \ldots, (2/(\tau_i-\tau_{i-1}))^{k-1}, 1, \ldots, (2/(\tau_i-\tau_{i-1}))^{k-1}\}$ for $i=1,\ldots,l$. We abbreviate $G_i(\tau_{i-1},\tau_i)$ as $G_i$ for $i=1\ldots,l$ and form $G(t):\mathbb{R}^{l+1}\rightarrow\mathbb{R}^{2kl\times 2kl}$ as $G(t)=\mathrm{diag}\{G_1,\ldots,G_l\}$. % The matrix $G(t)$ scales the rows of the $\widetilde{V}$ such that the constraints \eqref{eq::form2::vand} and \eqref{eq::form_nondim::vand} are equivalent. 

% The optimisation program \eqref{eq::form_nondim} is derived from \eqref{eq::form1} via a change of variables. %
While the optimal cost of \eqref{eq::form2} and \eqref{eq::form_nondim} are equal, the corresponding optimisers are different. The relationship between the solutions of the two programs is summarised next.

\begin{prop}\label{prop::scaling}
Let $a^\star$ and $f^\star$ solve \eqref{eq::form2} and $\sigma(\tau)$ be the spline with coefficients $a^\star$. Further, let $\widetilde{a}^\star=\mathrm{vec}\{\widetilde{a}^\star_1,\ldots, \widetilde{a}^\star_l\}$ such that $\widetilde{a}^\star_i$ is the vector of coefficients of $\widetilde{\sigma}_i(\rho)$ for $i=1,\ldots,l$. Then, $\widetilde{a}^\star$ and $f^\star$ solve \eqref{eq::form_nondim} if and only if $\widetilde{\sigma}_i(\rho)=\sigma_i((\tau_i-\tau_{i-1})\rho/2 + (\tau_i+\tau_{i-1})/2)$.
\end{prop}
\begin{proof}
See Appendix~\ref{app:scaling}.
\end{proof}

\section{Optimising set of times}\label{sec::time_optimising}

In this section we present another problem formulation for generating optimal spline trajectories. Closely related to \eqref{eq::form1}, in this problem the time between spline segments is included as a decision variable. We leverage the results of the previous section to propose an algorithm that efficiently solves this problem.

\subsection{Variable-time problem}

Motivated by the time-optimal minimum-snap trajectories of Mellinger and Kumar~\cite{Mellinger2011}, we formulate another trajectory generation problem in which the cost is also minimised with respect to the times $t$. We name this problem the variable-time problem.

% \begin{prob}[Variable-Time Problem]\label{prob2}
    
%     \begin{subequations} \label{eq::form3}
%     Let $t^\star$ be the minimising argument of the optimisation program
%         \begin{align}
%             \min_{t}\quad & J(t), \\
%             \text{s.t.} \quad & \tau_{i-1}<\tau_i, \quad i=1, \ldots,l,
%         \end{align}
%     \end{subequations}
%     and let $a^\star$ and $f^\star$ be the minimising arguments of \eqref{eq::form2} for $t=t^\star$. Find the continuous spline $\sigma(\tau)$ with coefficients $a^\star$.
% \end{prob}

\begin{prob}[Variable-Time Problem]\label{prob2}
    Let $t^\star=[\tau_0^\star,\ldots,\tau_l^\star]^T$ be the a minimiser of
    \begin{subequations} \label{eq::form3}
        \begin{align}
            \min_{t}\quad & J(t), \\
            \text{s.t.} \quad & \tau_{i-1}<\tau_i, \quad i=1, \ldots,l.
        \end{align}
    \end{subequations}
    Find the spline $\sigma(\tau)$ by solving the fixed-time problem as described in Problem~\ref{prob1} with $t=t^\star$.
\end{prob}

A common approach to solve \eqref{eq::form3} is a gradient descent method~\cite{Mellinger2011},~\cite{Richter2016},~\cite{burri2015real}. In these works, an approximation of the gradient $\nabla_t J(t)$ is used as a search direction. The calculation of the approximate gradient typically involves solving \eqref{eq::form2} repeatedly. We also employ a gradient descent method to solve the variable-time problem. The strength of our algorithm is that it requires solving \eqref{eq::form2} only once per iteration via solving a reformulation of \eqref{eq::form3}. This reformulation alongside the proposed solution method are introduced in the following subsection.

\subsection{Efficient gradient calculation}

\begin{algorithm}[t]
    \SetAlgoLined
    $d_i \leftarrow d_0$\;
    \While{not terminated}{
    Find $f^\star$ minimising \eqref{eq::form2} for $t=Rd_i$ using Algorithm~\ref{algo1}\;
    Calculate $\nabla_d J(Rd_i)$ using Proposition~\ref{prop::gradient}\;
    Choose $\alpha_i$\;
    $d_i\leftarrow {\bf P}_+\big(d_i-\alpha_i\nabla_d J(Rd_i)\big)$\;
    $i\leftarrow i+1$\;
    }
 \caption{Algorithm for Solving Problem~\ref{prob2}} \label{algo2}
\end{algorithm}
 
% We first introduce some notation required for the solution. We write the minimum cost of \eqref{eq::form3} for fixed $T$ as
% \begin{align*}
%     J(T)&= \min_{\widehat{a}, f} \quad \frac{1}{2}\widehat{a}^T \widehat{H}(T) \widehat{a}.
% \end{align*}
% We use $\widehat{a}^\star$ and $f^\star$ to denote the arguments that minimise \eqref{eq::form3}.

We begin by introducing the change of variables $\delta_i=(\tau_i-\tau_{i-1})/2$ for $i=1,\ldots,l$. Let $d=[\delta_1,\ldots,\delta_l]^T\in\mathbb{R}^l$. Further, let $R\in\mathbb{R}^{(l+1)\times l}$ be such that $R_{i,j} = 2$ if $i<j$ and zero otherwise. By construction, the span of $R$ is all $t\in\mathbb{R}^{l+1}$ such that $\tau_0=0$. The reformulation is then
\begin{subequations} \label{eq::form6}
    \begin{align}
        \min_{d}\quad & J(Rd), \\
        \text{s.t.} \quad & \delta_{i}>0, \quad i=1, \ldots,l, \label{eq::form6::const}
    \end{align}
\end{subequations}

We note that \eqref{eq::form6} is not strictly equivalent to \eqref{eq::form3} due to the different dimensions of $d$ and $t$. However, the next proposition states that the optimal cost $J(t)$ is invariant to the initial time $\tau_0$, and thereby every solution to \eqref{eq::form6} can be used to calculate another solution to \eqref{eq::form3}.

% \declan[inline]{Go through entire paper and fix notation.}

\begin{prop}\label{prop::variable}
The vector $d^\star$ is a local minimiser of \eqref{eq::form6} if and only if $t^\star = Rd^\star$ is a local minimiser of \eqref{eq::form3}.
\end{prop}
\begin{proof}
See Appendix~\ref{app::variable}.
\end{proof}

We solve \eqref{eq::form6} using a projected steepest descent method that we refer to as Algorithm~\ref{algo2}. The iteration of Algorithm~\ref{algo2} is
\begin{align}\label{eq::iteration}
    d_{i+1}=\mathbf{P}_+\big(d_i-\alpha_i \nabla_d J(R d_i)\big),
\end{align}
where the $d_i\in\mathbb{R}^{l}$ are the sequence of iterates and $\alpha_i\in\mathbb{R}$ are step sizes as chosen by a backtracking line search algorithm for integer $i\geq1$. We note that every iteration calculated with \eqref{eq::iteration} yields a feasible solution to \eqref{eq::form6} because the projection $\mathbf{P}_+$ renders the components of $d_{i+1}$ positive.
% \iman[inline]{Where is the projection step? Also, in the algorithms, it is undefined. % Declan: It's defined in the notation. }

Since $J(t)$ is the optimal cost of a QP with equality constraints, it has a closed form expression~\cite{sun2006optimization} that we use to calculate  $\nabla_d J(Rd)$ directly. We then exploit sparsity to derive expressions for $\nabla_d J(Rd)$ that are efficient in terms the computational cost of their calculation, which are presented in the following proposition. 

\begin{prop}\label{prop::gradient}
For a given $t$, let $a^\star$ and $f^\star=\mathrm{vec}\{f_1^\star, \ldots, f_l^\star\}$ solve \eqref{eq::form2}. The $i$th component of the Jacobian, $\big (\nabla_d J(Rd) \big )_i$, $i=1,\ldots,l$, is given by 
\begin{align}\label{eq::partialJ}
    & \big(\nabla_d J(Rd))_i = (2k-1)\delta_i^{1-2k}\big(\delta_i^{-1}{f_i^\star}^T H_i(-1,1) f_i^\star \\
    & \quad -{f_i^\star}^T H_i(-1,1) F(\delta_i) f_i^\star -(f_i^\star-\overline{f}_i)^T H_i(-1,1)(f_i^\star-\overline{f}_i) \nonumber \\
    & \quad -(f_i^\star-\overline{f}_i)^T \big( H_i(-1,1) F(\delta_i) + F(\delta_i) H_i(-1,1) \big) \overline{f}_i\big), \nonumber
\end{align}
with $F(\delta)= \delta^{-1}\mathrm{diag}\{0, 1, \ldots,k-1, 0, 1, \ldots,k-1\}\in\mathbb{R}^{2k\times 2k}$.
\end{prop}
\begin{proof}
See Appendix~\ref{app::gradient}.
\end{proof}

The component-wise expressions of \eqref{eq::partialJ} involve matrices of size at most $n \times n$, hence the calculation of $\nabla_d J(Rd)$ has computational complexity $O(ln^3)$. The minimising argument $f^\star$ is required to calculate \eqref{eq::partialJ}, which in itself requires $O(ln^3)$ computations. Hence, the computational complexity of calculating each iteration \eqref{eq::iteration} is $O(ln^3)$.

We also note because \eqref{eq::partialJ} was derived from \eqref{eq::form3}, it depends only on the difference between times $\delta_i$. As found in a previous work~\cite{burke2020}, the condition number of the matrices in \eqref{eq::partialJ} will not necessarily increase for large values of the times $\tau_i$. Instead, \eqref{eq::partialJ} is only prone to numerical errors for large values of $\delta_i$. Indeed, deriving similar expressions for $\nabla_d J(Rd)$ from \eqref{eq::form2} results in calculations involving ill-conditioned matrices that render the results practically useless due to numerical errors. 

\subsection{Numerical experiments}

\begin{table}[t]
    \centering
    \ra{1.3}
    \begin{tabular}{rrrr}\toprule
    & $l=6$ & $l=8$ & $l=10$ \\ \midrule
    Algorithm~\ref{algo2} & 2.31 & 2.48 & 2.45 \\
    Iteration \eqref{eq::zerod_approx} & 12.43 & 17.59 & 24.50 \\
    Iteration \eqref{eq::random_approx} & 39.52 & 232.74 & 996.70 \\
    \bottomrule \\
    \end{tabular} 
    \caption{Average number of $J$ evaluations for three methods over $100$ trials.} \label{tab::experiments}
\end{table}

In this subsection, we demonstrate another computational benefit of employing expression \eqref{eq::partialJ}. The benefit we look to highlight is in terms of the number of $J(t)$ evaluations required by a descent method to converge to a minimum. We compare our solution to two other methods popular in the literature by solving the variable-time problem and generating some minimum snap trajectories.

The first method we use for comparison is proposed by Mellinger and Kumar~\cite{Mellinger2011} and is chosen as it is easy to implement and is a conceptually simple derivative-free method. It uses a finite difference approximation of the gradient through the iteration, for $j\in\{1,\ldots,l\}$, 
\begin{align} \label{eq::zerod_approx}
    (d_{i+1})_j&=(d_i)_j-\beta_i \frac{J\big(R(d_i+\gamma_i e_j)\big)-J(Rd_i)}{\gamma_i}, 
\end{align}
where $\beta_i\in\mathbb{R}$ and $\gamma_i\in\mathbb{R}$ are sequences of parameters for integer $i\geq 1$. The vectors $e_j\in\mathbb{R}^{l}$ are such that $(e_j)_j=1$ and zero otherwise. The second method we use for comparison requires only one $J(t)$ evaluation per iteration. It is a random, derivative-free method with the iteration 
\begin{align}\label{eq::random_approx}
    d_{i+1}&=d_i-\varepsilon_i \frac{J\big(R(d_i+\zeta_i r)\big)-J(Rd_i)}{\zeta_i} r,
\end{align}
where $r\in\mathbb{R}^{l}$ is a random vector with a known Gaussian distribution and $\varepsilon_i\in\mathbb{R}$ and $\zeta_i\in\mathbb{R}$ are sequences of parameters. For for an analysis of the properties of random derivative-free methods and how to choose parameters such as the Gaussian distribution, we refer the reader to~\cite{nesterov2017random}.

% It scales poorly, however, and one evaluation of $J(t)$ is required for each dimension of the solution, for a total of $l+1$ $J(t)$ evaluations per iteration. We choose the approximation \eqref{eq::random_approx} as another point of comparison because it requires only one $J(t)$ evaluation per iteration.

The average results of $100$ trails, comparing the number of $J(t)$ evaluations required by three methods using the iterate updates \eqref{eq::iteration}, \eqref{eq::zerod_approx} and \eqref{eq::random_approx}, are summarised in Table~\ref{tab::experiments}. 
% \iman[inline]{How many experiments, are these numbers average numbers (I assume yes). Be precise.}
As expected, the exact gradient \eqref{eq::iteration} yields the least number of  $J(t)$ evaluations. The method that uses the finite-difference approximation \eqref{eq::zerod_approx} requires a similar number of iterations to that using the exact gradient \eqref{eq::iteration} but makes more $J(t)$ evaluations per iteration, evident from the experimental data that linearly increases in $l$. The expression \eqref{eq::random_approx} only requires one $J(t)$ evaluation per iteration, however, the total number of iterations required for convergence in using this update scales with the dimension of the solution and thus $l$. 

In summary, using the exact gradient in the update \eqref{eq::iteration} reduces the number of $J(t)$ evaluations through the calculation of each iterate and through the total number of iterations required for its convergence.

\section{Trajectory planning applications}\label{sec::experiments}

In this section, we present two applications of the algorithms presented thus far. The first is planning minimum snap trajectories for quadrotors and we demonstrate the real-time capability of our method by replanning a trajectory during an experimental flight. Second, we present a higher-level path planning algorithm that uses an RRT* algorithm to natively construct a tree of minimum snap trajectories. The computational complexity and well-conditioned formulation of our methods are critical to achieving the outcomes of each application.

\subsection{Minimum snap trajectory planning for quadrotors}

Our first application is an example of how our problem formulation and algorithms can be applied to Mellinger and Kumar's quadrotor trajectory planning methodology \cite{Mellinger2011}. This is achieved by adapting the variable-time problem to plan trajectories in multiple dimensions. We also demonstrate that our method can generate quadrotor trajectories in real-time.

We output trajectories in terms of the quadrotor's position $[\sigma_x(\tau), \sigma_y(\tau), \sigma_z(\tau)]^T\in\mathbb{R}^3$ and yaw $\sigma_\psi(\tau)\in[-\pi,\pi]$, such that $\sigma_x(\tau)$ and $\sigma_y(\tau)$ are two minimum snap splines and $\sigma_z(\tau)$ and $\sigma_\psi(\tau)$ are two minimum acceleration splines\footnote{This has the approximate effect of minimising the control effort required to track the trajectory (see \cite{Mellinger2011} for details).}. 
% \declan[inline]{Clarify $z$ is altitude and $x$ and $y$ are planar.}
We generate each spline by solving separate instances of the fixed-time problem with different parametrisations. To this end, we introduce the optimal costs of the four optimisation programs associated with each dimension of the trajectory. Let $J_x(t)$ and $J_y(t)$ be the optimal cost of two instances of \eqref{eq::form2} with $k=5$ % , i.e., $J_x(t)$ and $J_y(t)$ are the total snap of two minimum snap trajectories. 
and let $J_z(t)$ and $J_\psi(t)$ be the optimal cost of two instances of \eqref{eq::form2} with $k=3$. 
% \declan[inline]{Clarify different parametrisations.}
We solve a variant of the variable-time problem to find $\sigma_x(\tau),~\sigma_y(\tau),~\sigma_z(\tau)$ and $\sigma_\psi(\tau)$ such that the time between segments is the same for each spline.

\begin{prob}[Multi-Dimensional Variable-Time Problem]\label{prob3}
    Let $t^\star$ be the minimising argument of the optimisation program
    \begin{subequations} \label{eq::form5}
        \begin{align}
            \min_{t}\quad & J_x(t)+J_y(t)+J_z(t)+J_\psi(t), \label{eq::form5::cost} \\
            \text{s.t.} \quad & \tau_{i-1}<\tau_i, \quad i=1, \ldots,l.
        \end{align}
    \end{subequations}
    Find the splines $\sigma_x(\tau)$ and $\sigma_y(\tau)$ by solving two instances of the fixed time problem with $t=t^\star$ and $k=5$, and the splines $\sigma_z(\tau)$ and $\sigma_\psi(\tau)$ by solving two instances of the fixed time problem with $t=t^\star$ and $k=3$.
\end{prob}

We solve the multi-dimensional variable-time problem using a gradient descent method adapted from the approach taken in Section~\ref{sec::time_optimising}. This method is only trivially different to Algorithm~\ref{algo2}. 

\begin{figure}[t]
\centering
{\includegraphics[trim={0cm 0cm 0cm 0cm}, width=0.90\linewidth]{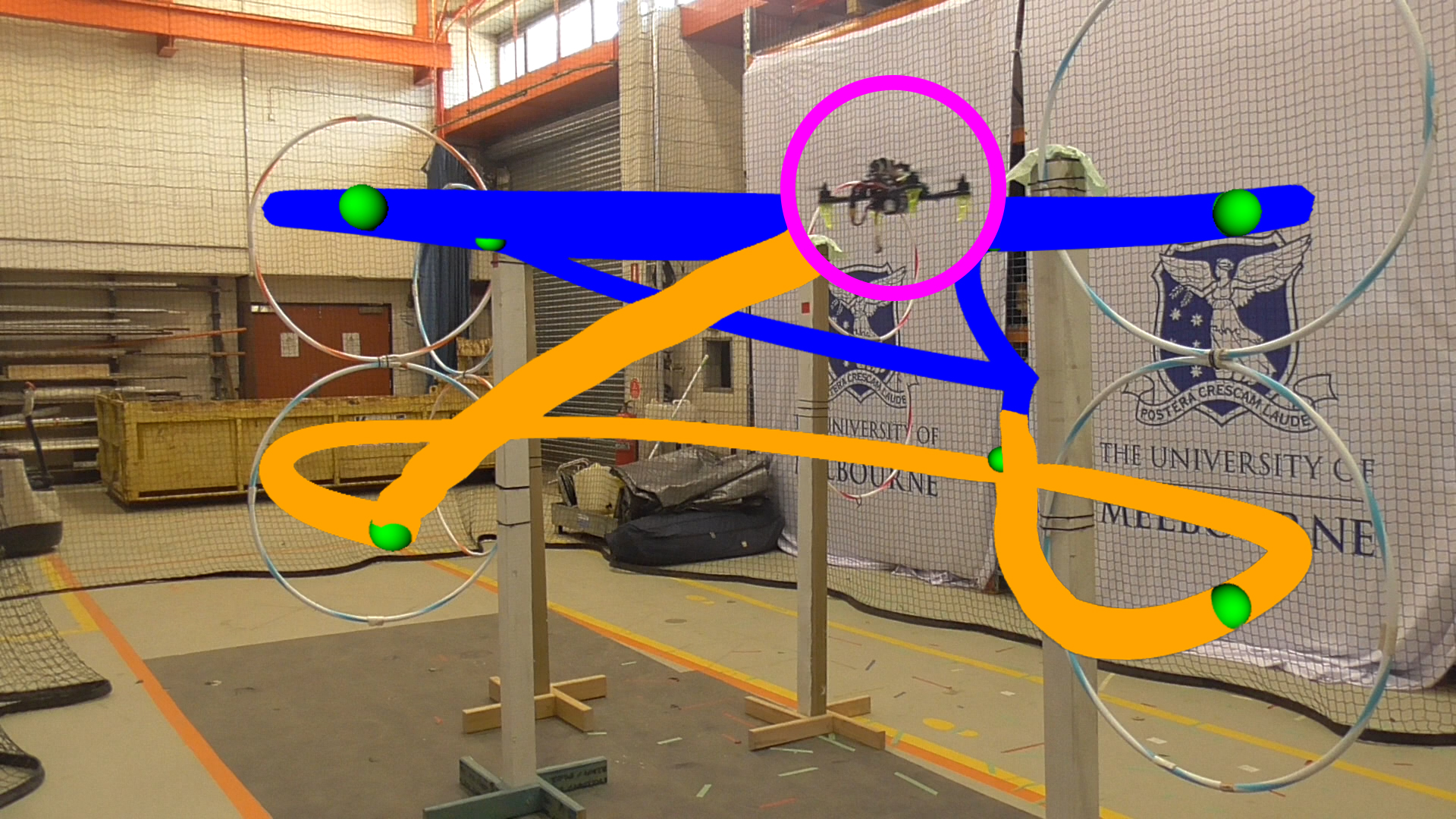}}
\caption{The quadrotor flight through an obstacle filled circuit. The flight has two phases: the offline trajectory (blue) and the recalculated, online trajectory (orange). \label{fig::quad_still}}
\end{figure}

\begin{figure}[t]
\centering
{\includegraphics[width=\linewidth]{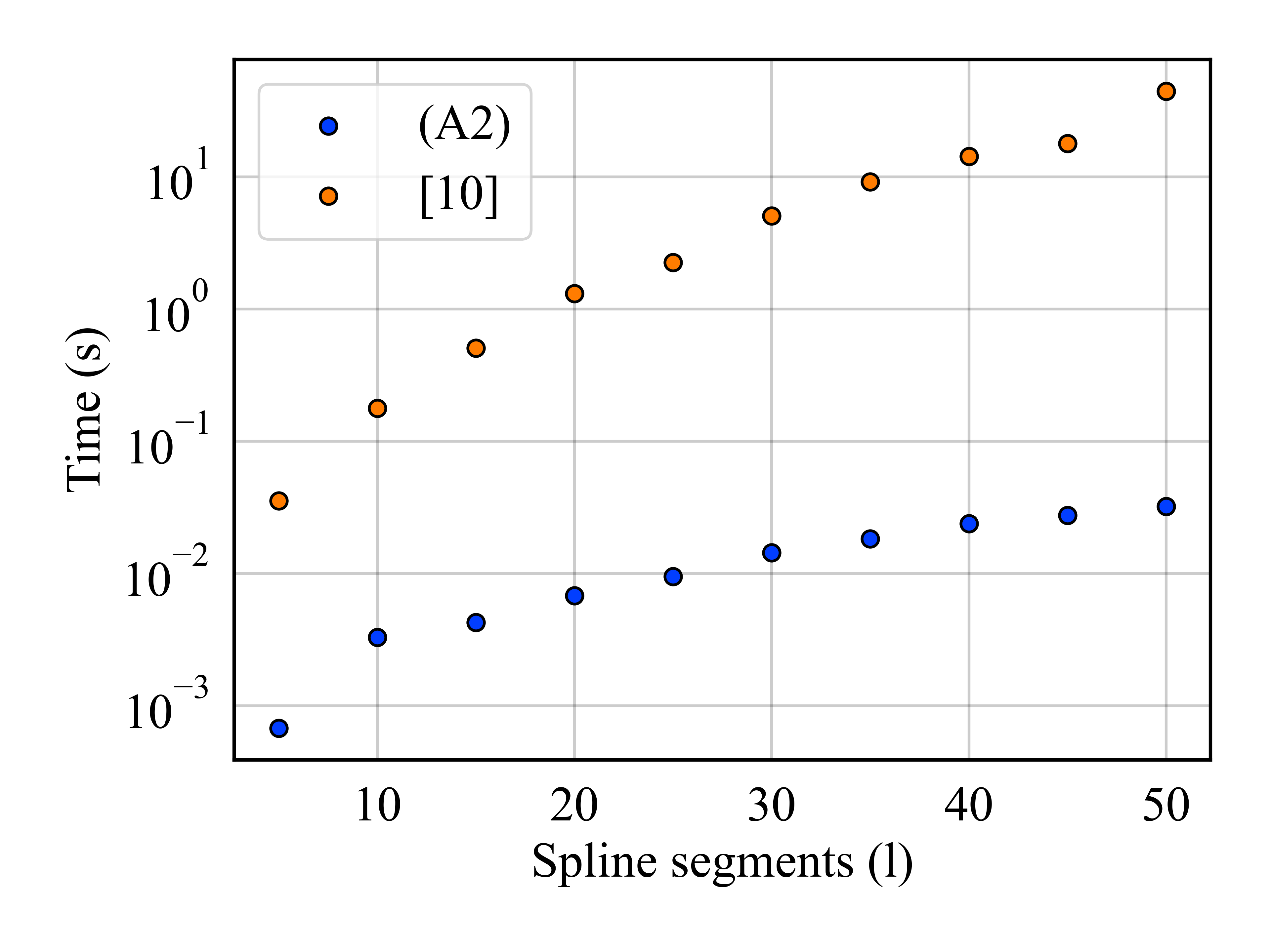}}
\caption{The time (s) taken to compute a range of trajectories by solving Problem~\ref{prob3} using an adapted Algorithm~\ref{algo2} (blue) and the method proposed by~\cite{Richter2016} as implemented in~\cite{eth_implementation} (orange). \label{fig::quad_timing} }
\end{figure}

\begin{figure}[h]
\centering
{\includegraphics[width=\linewidth]{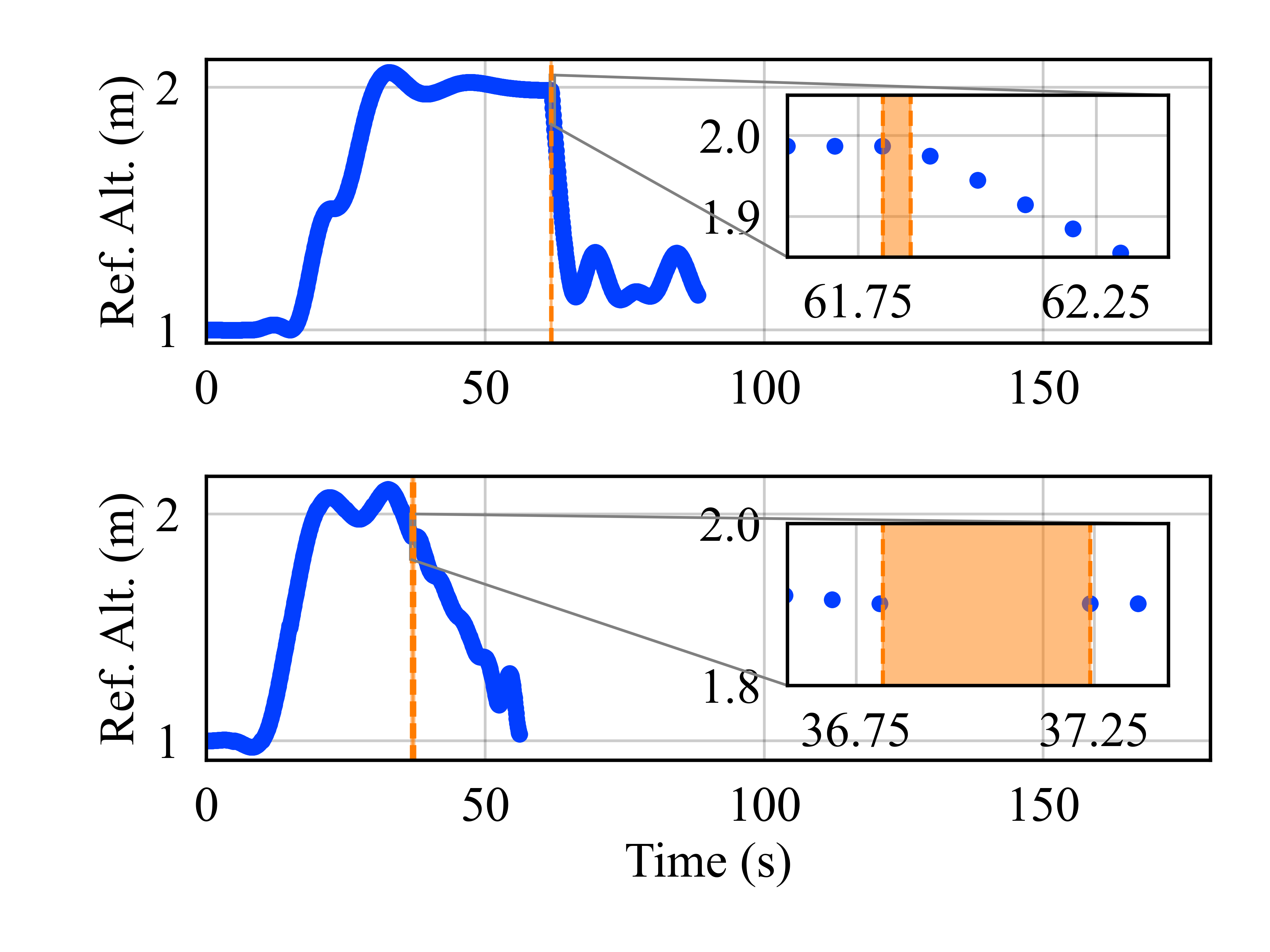}}
\caption{The discrete reference altitude of the quadrotor flight (blue) and the time taken to recalculate the online trajectory (orange). Using the adapted Algorithm~\ref{algo2} (top), the quadrotor tracks the offline trajectory through the upper hoops until $\tau=61.80$ s when it calculates a new trajectory in $75$ ms, before the next reference altitude update. Using the method proposed by~\cite{Richter2016} as implemented in~\cite{eth_implementation} (bottom), the new trajectory is calculated in $435$ ms, requiring the quadrotor to pause its flight. \label{fig::recalc}}
\end{figure}

To demonstrate the real-time capability of our algorithm, we conduct a quadrotor flight in which the trajectory is replanned midflight. The flight is through a circuit of four pairs of hoops stacked vertically on one another. Initially, the multi-dimensional variable-time problem is parametrised such that the quadrotor circumnavigates the course by passing through the upper hoops. Then, prompted by user input at an unspecified time, a new parametrisation of the multi-dimensional variable-time problem is solved to generate a trajectory that passes through the lower hoops. The two trajectories are visualised in the flight space in Fig.~\ref{fig::quad_still}. We benchmark our method by repeating the flight using Richter {\it et al.}'s solution of the variable-time problem~\cite{Richter2016} as implemented by the Autonomous Systems Lab of ETH Zurich~\cite{burri2015real}. This implementation is also a gradient descent method and we note that it utilises sparse matrix operations to perform calculations.

Trajectory generation was performed with C++ implementations of each algorithm and executed on a laptop with an Intel Core i7-8650U CPU running at 1.9 GHz, with 16 GB of RAM. The quadrotor we use for testing is assembled from generic components and a Pixhawk2 The Cube Black flight controller. The flight controller and a Vicon motion capture system are used for sensing capabilities. We use PX4 firmware~\cite{PX4DevTeam2020} for control and estimation. 

The computation time required for calculating a range of offline trajectories is shown in Fig.~\ref{fig::quad_timing}. Under our method, it took $0.43$ ms to generate a trajectory for one lap of the circuit and $5.3$ ms to generate one for five laps. Evidently Algorithm~\ref{algo2} is faster than the solver from~\cite{eth_implementation}, which generated a single-lap trajectory in $21$ ms and a five-lap trajectory in $0.91$ s. 

The timing of the trajectory replanning event is captured in Fig.~\ref{fig::recalc}, where the time taken to calculate the new trajectory is plotted in orange against the reference altitude updates in blue. The inset plot shows that, for each flight, the trajectory is recalculated immediately after the reference altitude is updated. Under our method, corresponding to Fig.~\ref{fig::recalc}, the new trajectory is calculated in approximately $75$ ms before the next reference position update. Thus the new trajectory is able to be broadcast without affecting the reference updates and quadrotor flight. In comparison, when using the solver from~\cite{eth_implementation}, the new trajectory is calculated in approximately $435$ ms. This exceeds the reference position update period and the quadrotor is forced to pause its flight. 

Finally, we note that the trajectories generated under the two methods are different (see Fig.~\ref{fig::recalc}). Additional interpolation constraints were employed under our method for safety purposes. The solver from~\cite{eth_implementation} was not able to compute a feasible trajectory with these additional constraints hence they were not included. In terms of the comparison, the additional constraints result in spline trajectories with more segments, thus increasing the computation times under our method and favouring to the benchmark.

\subsection{Minimum Snap Spline RRT* path planning}

Our second application is an example of how the methods presented thus far can be incorporated in an RRT* algorithm, to constructs minimum snap trajectories. This version version of the RRT* algorithm, titled \emph{Minimum Snap Spline RRT*}, is presented in Algorithm~\ref{algo3}
At each iteration of Algorithm~\ref{algo3}, a sample from free space is either added to the graph or discarded, depending on the minimum snap trajectory through the graph to the sample. We start with the problem setup and define the problem used to generate minimum snap trajectories, then we present Algorithm~\ref{algo3} and some simulation results.

To describe the problem space, let $\mathcal{X}\subseteq\mathbb{R}^2$ be the configuration space, $\mathcal{X}_\mathrm{obs}\subset\mathcal{X}$ be the obstacle space and $\mathcal{X}_\mathrm{free}=\mathcal{X}\backslash\mathcal{X}_\mathrm{obs}$ be the free space. We note that the problem space is formulated in two dimensions for simplicity but to generalise subsequent results to higher dimensions is trivial. Next, we introduce some graph theoretic notation. Let $\mathcal{G}=(\mathcal{V}, \mathcal{E})$ be a tree graph with vertices $\mathcal{V}=\{1,\ldots,n_g\}$ and edges $\mathcal{E}\subset\mathcal{V}\times\mathcal{V}$. A configuration is the collection of points $v=\mathrm{vec}\{v_1,\ldots,v_{n_g}\}\in\mathbb{R}^{2{n_g}}$ where $v_i\in\mathbb{R}^2$ for $i=1,\ldots,n_g$. A framework is the embedding of the graph $\mathcal{G}$ in two-dimensional space denoted by the pair $(\mathcal{G},v)$, that is, for each $i\in\mathcal{V}$ there exists a $v_i\in\mathbb{R}^2$. % Let $v_0\in\mathbb{R}^2$ be the root of $\mathcal{G}$.

% To describe the subproblem, let $i_h$ be a vertex in the framework $(\mathcal{G},v)$ and let $(i_1,\ldots,i_h)$ be the path from the root of the graph $i_1=1$ to $i_h$ with $i_j\in\mathcal{V}$, $j\in\{1,\dots,h\}$. Further, let $u\in \mathcal{X}_\mathrm{free}$ be a sample from free space. The subproblem outputs a trajectory in the two-dimensional position $[\sigma_1(\tau),\sigma_2(\tau)]^T\in\mathbb{R}^2$, where $\sigma_1(\tau)$ and $\sigma_2(\tau)$ are two minimum snap splines satisfying the interpolation constraints
% \begin{subequations} \label{eq::rrt_cons}
%     \begin{align}
%         [\sigma_1(\tau_{j}),\sigma_2(\tau_{j})]^T &= v_{i_{j+1}}, \quad j=0,\ldots,h-1, \\
%         [\sigma_1(\tau_h),\sigma_2(\tau_h)]^T &= u.
%     \end{align}
% \end{subequations}
% From the constraints \eqref{eq::rrt_cons} we follow the constructive process outlined in Section \ref{sec::formulation} to formulate two instances of \eqref{eq::form2}. In this way, the subproblem is completely characterised by the parameters $\{\mathcal{X}_\mathrm{free}, (\mathcal{G},v), i_h, u\}$.

To develop a minimum snap aware RRT* algorithm, the following problem needs to be solved for any additional sample point.

\begin{prob}[Minimum Snap RRT* Subproblem]\label{prob4}
    For given parameters $\{\mathcal{X}_\mathrm{free}, (\mathcal{G},v), i_h, u\}$, where  $\mathcal{X}_\mathrm{free}\subseteq\mathbb{R}^2$,
     $(\mathcal{G},v)$ is a framework,
      $i_h$ is a vertex in $\mathcal{G}$, and
        $u\in\mathcal{X}_\mathrm{free}$,
    let $(i_1,\ldots,i_h)$ be the path from the root of the graph $i_1=1$ to $i_h$ with $i_j\in\mathcal{V}$, $j\in\{1,\dots,h\}$. Consider two instances of the fixed-time problem with constraints 
    %% OPTION A: Separate constraints for two different optimisation programs
    \begin{align*}
        \sigma_1(\tau_{j}) &= (v_{i_{j+1}})_1, \quad j=0,\ldots,h-1, \\
        \sigma_1(\tau_l) &= (u)_1,
        \intertext{and}
        \sigma_2(\tau_{j}) &= (v_{i_{j+1}})_2, \quad j=0,\ldots,h-1, \\
        \sigma_2(\tau_l) &= (u)_2,
    \end{align*}
    %% OPTION B: (Combined) equations from which two sets of constraints are derived.
    % \begin{subequations} \label{eq::rrt_cons}
    %     \begin{align}
    %         [\sigma_1(\tau_{j}),\sigma_2(\tau_{j})]^T &= v_{i_{j+1}}, \quad j=0,\ldots,h-1, \\
    %         [\sigma_1(\tau_h),\sigma_2(\tau_h)]^T &= u.
    %     \end{align}
    % \end{subequations}
    and let $J_1(t)$ and $J_2(t)$ be their corresponding optimal costs, respectively. Let $t^\star$ be the minimising argument of the optimisation program
        \begin{align*}
            \min_{t}\quad & J_1(t)+J_2(t), \\
            \text{s.t.} \quad & \tau_{i-1}<\tau_i, \quad i=1, \ldots,h.
        \end{align*}
    Find the splines $\sigma_1(\tau)$ and $\sigma_2(\tau)$ by solving two instances of the fixed time problem with $t=t^\star$ and $k=5$.
\end{prob}

We note that from the constraints in Problem \ref{prob4} the constructive process outlined in Section \ref{sec::formulation} is used to formulate two instances of \eqref{eq::form2}. % In this way, the subproblem is completely characterised by the parameters $\{\mathcal{X}_\mathrm{free}, (\mathcal{G},v), i_h, u\}$.

Following from Problem~\ref{prob4}, we now define two procedures used in Algorithm~\ref{algo3}. To ease notation, we explicitly represent their functions in terms of $i_h$ and $u$ and note that $\mathcal{X}_\mathrm{free}$ and $(\mathcal{G},v)$ are implicit parameters. The procedures are as follows.

\begin{itemize}
    \item {\bf Collision checking:} For given parameters $\{\mathcal{X}_\mathrm{free}, (\mathcal{G},v), i_h, u\}$, the following Boolean function is defined
\end{itemize}
\begin{align*}
    \mathrm{CollisionFree}(i_h, u) &= \begin{cases}
    \mathrm{True} & \begin{aligned}
        & \text{if }[\sigma_1(\tau),\sigma_2(\tau)]^T\in\mathcal{X}_\mathrm{free}  \\
        & \text{ for all }\tau^\star_0\leq \tau \leq \tau^\star_l,
    \end{aligned} \\
    \mathrm{False} & \text{otherwise,}
    \end{cases}
\end{align*}
     where $\sigma_1(\tau)$, $\sigma_2(\tau)$ and $t^\star=[\tau_0^\star,\ldots,\tau_l^\star]^T$ are found by solving Problem~\ref{prob4} with $\{\mathcal{X}_\mathrm{free}, (\mathcal{G},v), i_h, u\}$.
\begin{itemize}
    \item {\bf Snap evaluation:} For given parameters $\{\mathcal{X}_\mathrm{free}, (\mathcal{G},v), i_h, u\}$, the following real valued function is defined
\end{itemize}
\begin{align*}
    \mathrm{Cost}(i_h, u) &= \int_{\tau^\star_0}^{\tau^\star_l} \big(\sigma^{(4)}_1(\tau)\big)^2 + \big(\sigma^{(4)}_2(\tau)\big)^2d\tau,
\end{align*}
     where $\sigma_1(\tau)$, $\sigma_2(\tau)$ and $t^\star=[\tau_0^\star,\ldots,\tau_l^\star]^T$ are found by solving Problem~\ref{prob4} with $\{\mathcal{X}_\mathrm{free}, (\mathcal{G},v), i_h, u\}$.

We are now ready to present Algorithm~\ref{algo3}, which is based on the RRT* algorithm implementation in \cite{karaman2011sampling}. The algorithm is initialised with a single vertex $v_1$ and proceeds to construct a framework. At each iteration, the algorithm performs the following coarse steps:

\begin{algorithm}[t]
\SetAlgoLined
$\mathcal{V}\leftarrow \{1\};~\mathcal{E}\leftarrow \varnothing;~v\leftarrow v_1$\;
\While{not terminated}{
    Sample $v_\mathrm{rand}$ from $\mathcal{X}_\mathrm{free}$\; \label{alg::sample}
    $i_\mathrm{init}\leftarrow \{i\in\mathcal{V} | \arg \min_{i} \lVert v_i-v_\mathrm{rand}\rVert_2$\}\; \label{alg::nearest}
    \If{$\mathrm{CollisionFree}(i_\mathrm{init}, v_\mathrm{rand})$}{ \label{alg::nearest_check}
        $\mathcal{V}\leftarrow \mathcal{V}\cup\{\lvert \mathcal{V}\rvert +1 \}$\; \label{alg::add_v_1}
        $v\leftarrow \mathrm{vec}\{v, v_\mathrm{rand}\}$\; \label{alg::add_v_2}
        $i_\mathrm{min}\leftarrow i_\mathrm{init}$\;
        $J_\mathrm{min}\leftarrow \mathrm{Cost}(i_\mathrm{init}, v_\mathrm{rand})$\;
        $\mathcal{V}_\mathrm{near}\leftarrow\{i\in\mathcal{V}|\lVert v_i - v_\mathrm{rand} \rVert_2 < \varepsilon\}$\;
        % Find $\mathcal{V}_\mathrm{near}\subset\mathcal{V}$ within some radius of $v_\mathrm{rand}$\;
        \For{$i_\mathrm{near}\in \mathcal{V}_\mathrm{near}$}{ \label{alg::near_1}
         \If{$\mathrm{CollisionFree}(i_\mathrm{near}, v_\mathrm{rand})$ and  ($\mathrm{Cost}( i_\mathrm{near}, v_\mathrm{rand}) < J_\mathrm{min}$) \label{alg::all_v_2}} {
                $i_\mathrm{min}\leftarrow i_\mathrm{near}$\;
                $J_\mathrm{min}\leftarrow \mathrm{Cost}( i_\mathrm{near}, v_\mathrm{rand})$\;
                }
            } \label{alg::near_2}
        % $\mathcal{E}\leftarrow \mathcal{E}\cup\{(v_\mathrm{min},v_\mathrm{rand})\}$\;
        $\mathcal{E}\leftarrow \mathcal{E}\cup\{(i_\mathrm{min},\lvert \mathcal{V} \rvert)\}$\; \label{alg::edge}
        \For{$i_\mathrm{near}\in \mathcal{V}_\mathrm{near}$}{ \label{alg::rewire_1}
            $i_\mathrm{par}\leftarrow \{i\in\mathcal{V} | (i,i_\mathrm{near})\in\mathcal{E}\}$\;
            % $v_\mathrm{parent}\leftarrow \{v | (v, v_\mathrm{near})\in\mathcal{E}\}$\;
            $J_\mathrm{near}\leftarrow \mathrm{Cost}(i_\mathrm{par}, v_{i_\mathrm{near}})$\;
            \If {$\mathrm{CollisionFree}(\lvert \mathcal{V} \rvert, v_{i_\mathrm{near}})$ and ($\mathrm{Cost}( \lvert \mathcal{V} \rvert, v_{i_\mathrm{near}}) < J_\mathrm{near}$) \label{alg::all_v_3}} {
                $\mathcal{E}\leftarrow \mathcal{E}\backslash\{(i_\mathrm{par}, i_\mathrm{near})\}$\; 
                $\mathcal{E}\leftarrow \mathcal{E}\cup\{(\lvert \mathcal{V} \rvert, i_\mathrm{near})\}$\; \label{alg::rewire_2}
            }
        }
    }
}
Return $((\mathcal{V},\mathcal{E}),v)$\;
\caption{Minimum Snap Spline RRT*} \label{algo3}
\end{algorithm}
\begin{figure}[h!]
\centering
{\includegraphics[width=0.98\linewidth]{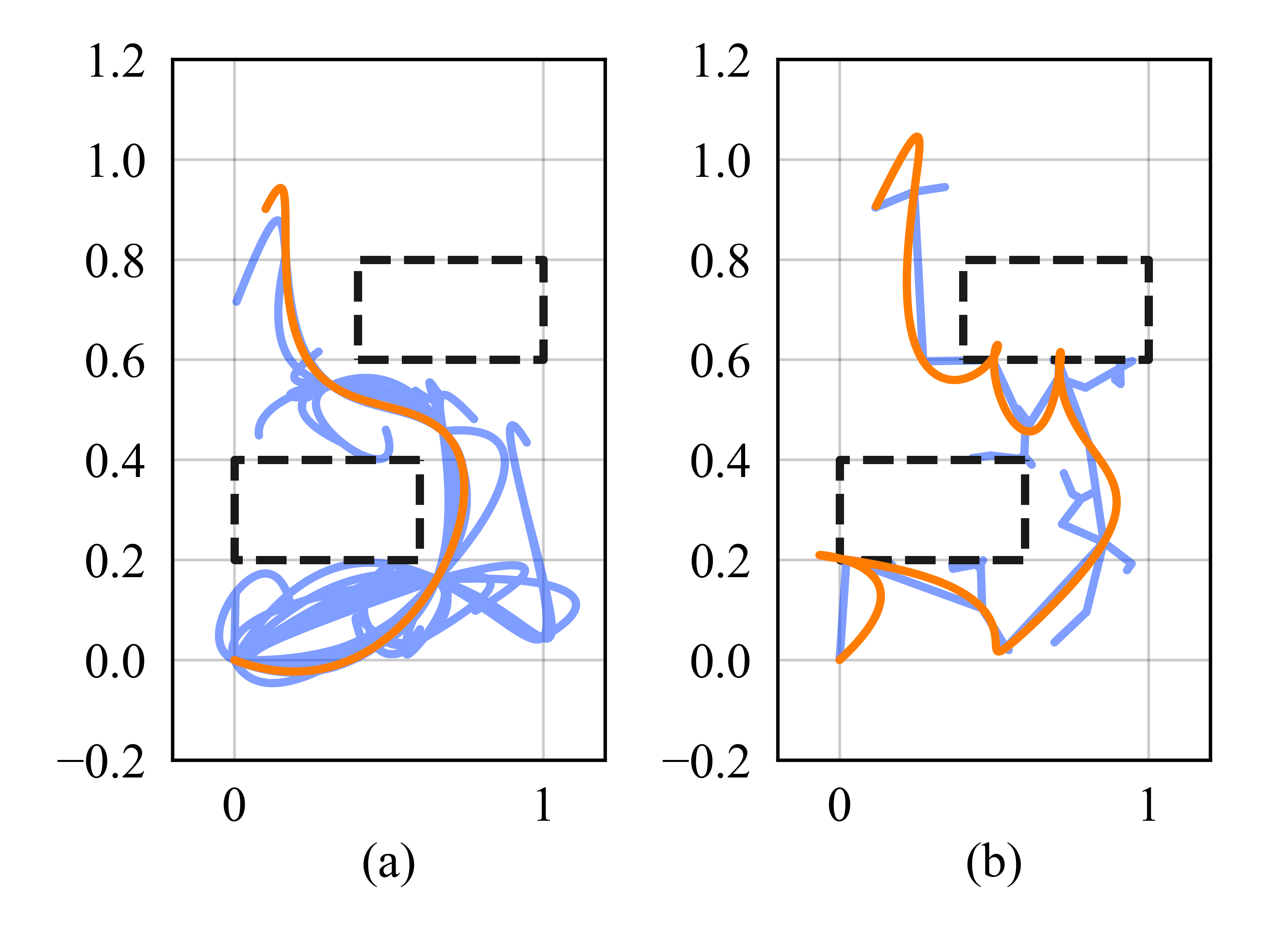}}
\caption{Minimum snap trajectories generated using two RRT* algorithms: (a) path cost is the total snap of a minimum snap spline through vertices; and (b) path cost is Euclidean distance between vertices, which is used to generate a minimum snap spline. Pictured is the underlying tree (solid blue), final trajectory (solid orange) and obstacles (dashed black).}
\label{fig::rrt} 
\end{figure}

\begin{enumerate}
    \item {\bf Sample:} A point $v_\mathrm{rand}$ is randomly sampled (Line \ref{alg::sample}).
    \item {\bf Add sample to framework:} Find $i_\mathrm{init}$, the vertex corresponding to the nearest point in the framework to $v_\mathrm{rand}$ (Line \ref{alg::nearest}). If a collision-free path exists from the root $v_1$ to $v_\mathrm{rand}$ through $v_{i_\mathrm{init}}$, then $v_\mathrm{rand}$ is added to the framework with the corresponding vertex $i_{\lvert \mathcal{V} \rvert}$ (Lines \ref{alg::nearest_check}--\ref{alg::add_v_2}).
    \item {\bf Add edge to sample:} Compare the minimum snap trajectories from the root $v_1$ to $v_\mathrm{rand}$ through vertices corresponding to points near the sample $\mathcal{V}_\mathrm{near}$ (Lines \ref{alg::near_1}--\ref{alg::near_2}). Find the point $i_\mathrm{min}\in\mathcal{V}_\mathrm{near}$ corresponding to the lowest snap trajectory and add an edge between $i_\mathrm{min}$ and the sample $i_{\lvert \mathcal{V} \rvert}$ (Line \ref{alg::edge}).
    \item {\bf Rewire edges:} Check if lower snap trajectories exist through the sample $v_\mathrm{rand}$ to nearby points in the framework then add and remove edges if so (Lines \ref{alg::rewire_1}--\ref{alg::rewire_2}).
\end{enumerate}

\begin{table}[t]
    \centering
    \ra{1.3}
    \begin{tabular}{rrr} \toprule
    & Total Snap & Computation Time (s)  \\ \midrule
    Algorithm~\ref{algo3} & 14 & 2.75 \\
    ED RRT* & 23 & 0.0123 \\
    \bottomrule \\
    \end{tabular} 
    \caption{Average total snap and average total computation time from a Monte Carlo simulation of $100$ trials of Algorithm~\ref{algo3} and the Euclidean distance (ED) RRT* algorithm.} \label{tab::rrt_metrics}
\end{table}

To provide some insight about the efficacy of Algorithm~\ref{algo3}, we generate minimum snap trajectories through one problem space. As a point of comparison, we also use the output of a Euclidean distance RRT* algorithm to create a minimum snap trajectory, which is a prototypical path planning approach~\cite{Richter2016},~\cite{burri2015real}. The Euclidean distance RRT* algorithm constructs a framework based on the straight-line path between vertices, that is, collision checking is performed over straight-line paths and the cost function is Euclidean distance. The trajectories generated under the two approaches are illustrated in Fig.~\ref{fig::rrt}, where the underlying tree is coloured in blue and the final trajectory in orange. Results from a Monte Carlo simulation of $100$ trials of the two methods is recorded in Table~\ref{tab::rrt_metrics}. The computational complexity of calculating the total snap of a trajectory is much greater than the Euclidean distance, which is reflected in the computation time of each method. However, this appears to be in a trade off with the total snap of the trajectory as Algorithm~\ref{algo3} generated trajectories with smaller total snap than the benchmark.

We conclude by highlighting a desirable feature of Algorithm~\ref{algo3}: the trajectories output by the algorithm inherit the properties of the collision checker. This is not the case for existing path planning methods that employ RRT* algorithms. For example, the Euclidean distance RRT* algorithm used as a benchmark does not take into account minimum snap trajectories and hence its output is prone to collisions (e.g., see Fig.~\ref{fig::rrt}). Algorithm~\ref{algo3}, however, executes the collision checking subroutine at each instance of adding a vertex or edge. Thus, the output trajectory is guaranteed to be collision free up to the properties of the collision checking subroutine. For example, we have assumed that $\mathrm{CollisionFree}(i_h,u)$ is a perfect collision checker thereby the output of Algorithm~\ref{algo3} is guaranteed to be collision free. % \footnote{For a more thorough consideration of collision checking, the authors recommend \cite{ImansPaper}.}. 
In this way, our method provides a native way of generating minimum snap trajectories.

% \declan[inline]{Discussion about collision free.
% % Assume perfect collision checker. Existing methods don't take into account the minimum snap trajectory, ours inherits the properties of the collision checker.
% % Our method is a native way of generating minimum snap trajectories.
% }

\section{Conclusion}
In this paper we present a general framework for creating optimal splines. We propose a well-conditioned formulation for an optimal spline generation problem and solve it using an algorithm with linear computational complexity in the number of segments in the spline trajectory. We leverage this algorithm to present a solution to another the optimisation problem used to generate optimal splines. We present expressions that reduce the computational complexity of optimizing the time between segments in the spline trajectory. We demonstrate the applicability of this general framework experimentally by generating trajectories for a quadrotor. We also highlight how our algorithms can be utilised as part of other trajectory planning approaches by proposing an RRT* algorithm that constructs trees of minimum snap splines.

In future work, we will explore the use of B-Splines for trajectory planning. The matrices that result from formulating constraints on B-Spline trajectories are typically banded and could permit further improvements in the computational complexity of solving the KKT conditions. 

\appendices
\section{Proof of Proposition~\ref{prop::alg_complexity}} \label{app::alg_complexity}

    % \begin{figure*}[!t]
    % 	\normalsize
    % 	%        \setcounter{MYtempeqncnt}{\value{equation}}
    % 	% Set the equation number to one less than the one
    % 	% desired for the first equation here.
    % 	% The value here will have to changed if equations
    % 	% are added or removed prior to the place these
    % 	% equations are referenced in the main text.
    % 	% \setcounter{equation}{5}
    	
    % 	\begin{align}
    % 	\begin{bmatrix}
    % 	D_1 &  -M^T & & & \\
    % 	-M & D_2 & -M^T & & \\
    % 	& -M & D_3& -M^T  &  \\
    % 	& & & \ddots & \\
    % 	& & & -M & D_l
    % 	\end{bmatrix}&=\begin{bmatrix}
    % 	I &   & &  & \\
    % 	N_1 & I & & & \\
    % 	& N_2 & I & & \\
    % 	& & & \ddots & & \\
    % 	& & & N_{k-1} & I
    % 	\end{bmatrix}\begin{bmatrix}
    % 	D_1 &  -M^T & & & \\
    % 	& \overline{D}_2 & -M^T & & \\
    % 	& & \overline{D}_3 & -M^T & \\
    % 	& & & \ddots & & \\
    % 	& & & & \overline{D}_l
    % 	\end{bmatrix} \label{eq::LU}
    % 	\end{align}
    % 	% % Restore the current equation number.
    % 	% \setcounter{equation}{\value{MYtempeqncnt}}
    % 	% IEEE uses as a separator
    % 	% \vspace{1em}
    % 	\hrulefill
    % 	% The spacer can be tweaked to stop underfull vboxes.
    % 	% \vspace*{1em}
    % \end{figure*}
    
    In this appendix we prove that Algorithm~\ref{algo1} solves the fixed-time problem. To reduce notation in this section, we will not explicitly state time dependence, i.e., we write $V(t),~H(t)$ as $V,~H$. We note that the fixed-time problem is solved for a known, constant $t$.
	
	We first present a useful set of recursions in the following lemma. The recursions efficiently solve a structured matrix equation that we will encounter in the principle proof.

    \begin{lemm}\label{lemm::recursions}
        For $i=1\ldots,k$, let $g_i^-$, $g_i^+$ and $\lambda_i$ satisfy 
        \vspace{-1em}
            \begin{subequations}\label{recs}
                \begin{align}
                g_i^-
                &=
                \overline{A}^{-1}_i(\overline{c}_{i}^- -B_ig_{i}^+), \label{rec1} \\
                g_i^+
                &=\begin{cases}
                \begin{aligned}
                & (C_l-B_l^T\overline{A}_l^{-1}B_l)^{-1}  \\
                & \quad (g_{l}^+-B_l^T\overline{A}_l^{-1}\overline{c}_{l}^-)
                \end{aligned} 
                & i=l, \\
                g_{i+1}^- & i=l-1, \ldots,1,
                \end{cases} \label{rec2} \\
                \lambda_i &= B_i^T g_{i}^- + C_i g_{i}^+ -c_i^+, \label{rec3}
                \end{align}
            \end{subequations}
            where
            \begin{subequations}
                \begin{align}
                \overline{A}_i &=\begin{cases}
                A_1 & i=1, \\
                A_{i}+C_{i-1}-B_{i-1}^T\overline{A}_{i-1}^{-1}B_{i-1} & i=2,\ldots,k,
                \end{cases} \label{con1}\\
                \overline{c}_{i}^-&=\begin{cases}
                {c}_{1}^- & i=1, \\
                \begin{aligned}
                & c_{i}^- + c_{i-1}^+ \\
                & \quad -B_{i-1}^T\overline{A}_{i-1}^{-1}\overline{c}_{i-1}^-
                \end{aligned}
                 & i=2,\ldots,k.
                \end{cases} \label{con2}
                \end{align} \label{con1and2}
            \end{subequations}
            Then \eqref{recs} solves 
            \begin{align}\label{eq::KKT}
            \begin{bmatrix}
            Z^TV^{-T}HV^{-1}Z & Z^TE^T \\
            EZ & 0
            \end{bmatrix}\begin{bmatrix}
            g \\ \lambda
            \end{bmatrix} &= \begin{bmatrix}
            Z^TV^{-T}HV^{-1}\overline{f} \\ 0
            \end{bmatrix},
        \end{align}
        where $\lambda = [\lambda_1^T, \dots, \lambda_{l-1}^T]^T\in\mathbb{R}^{k(l-1)}$ with $\lambda_i\in\mathbb{R}^k$ for $i=1,\ldots,l-1$.
        \end{lemm}
        % \vspace{1em}
        \begin{proof}
            Under the partition \eqref{eq::partition}, we may permute the variables and columns of \eqref{eq::KKT} to reveal the block-tridiagonal structure
        \begin{align}\label{tridiag}
            \underbrace{ \begin{bmatrix}
            D_1 &  -M^T & & \\
            -M & D_2 & -M^T & \\
            & -M & D_3 & \hdots\\
            & & \vdots & \ddots
            \end{bmatrix}}_{D}\begin{bmatrix}
            y_1 \\ y_2 \\ y_3 \\ \vdots 
            \end{bmatrix} &= \begin{bmatrix}
            d_1 \\ d_2 \\ d_3 \\ \vdots 
            \end{bmatrix}, 
            \end{align}
            where 
            \begin{align*}
                D_i=\begin{bmatrix}
                A_i & B_i & 0 \\
                B_i^T & C_i & I \\
                0 & I & 0
                \end{bmatrix}, \quad
                M= \begin{bmatrix}
                0 & 0 & I \\
                0 & 0 & 0 \\
                0 & 0 & 0
                \end{bmatrix},
        \end{align*}
        with $y_i=[{g_{i}^-}^T, {g_{i}^+}^T, \lambda_i^T]^T$ and $d_i=[{c_{i}^-}^T, {c_{i}^+}^T, 0]^T$ for $i=1,\ldots,l-1$ and $y_l=[{g_{l}^-}^T, {g_{l}^+}^T]^T$ and $d_l=[{c_{l}^-}^T, {c_{l}^+}^T]^T$ such that  $y=\mathrm{vec}\{y_1,\ldots,y_l\}$ and $d=\mathrm{vec}\{d_1,\ldots,d_l\}$.
        
        Block-tridiagonal matrices such as \eqref{tridiag} are commonly solved through Block LU factorisation~\cite{Fox1969}. The LU factorisation of the block tridiagonal matrix $D$ 
        \scalebox{0.85}{\parbox{1.1\linewidth}{%
			\begin{align*}
                D&=\underbrace{\begin{bmatrix}
            	I &   & & & \\
            	N_1 & I & & & \\
            	& N_2 & I & & \\
            	& & & \ddots & & \\
            	& & & N_{k-1} & I
            	\end{bmatrix}}_{L}\underbrace{\begin{bmatrix}
            	D_1 &  -M^T & & \\
            	& \overline{D}_2 & -M^T & & \\
            	& & \overline{D}_3 & & \\
            	& & & \ddots & & \\
            	& & & & \overline{D}_l
            	\end{bmatrix}}_{U}, \label{eq::LU}
            \end{align*}
		}}
        where the block elements are
        \begin{align*}
            N_i&=\begin{bmatrix}
            B_i^T\overline{A}_i^{-1} & -I  & C_i-B_i^T\overline{A}_i^{-1}B_i \\
            0 & 0 & 0  \\
            0 & 0 & 0  
            \end{bmatrix},\quad i=1,\ldots,l-1, \\
            \overline{D}_i&=\begin{bmatrix}
            A_i+C_{i-1}-B_{i-1}^T\overline{A}_{i-1}^{-1}B_{i-1} & B_i & 0 \\
            B_i^T & C_i & I \\
            0 & I & 0
            \end{bmatrix},\quad i=2,\ldots,l.
        \end{align*}
        First solving $Lx=d$ yields the iteration over $\overline{c}_{i}^-$, that is \eqref{con2}. Calculating each $A_i+C_{i-1}-B_{i-1}^T\overline{A}_{i-1}^{-1}B_{i-1}$ gives rise to \eqref{con1}. These matrices are then used to solving $Uy=x$, providing expressions for $g_{i}^-$ and $g_{i}^+$ as \eqref{rec1} and \eqref{rec2}. The solution to $Uy=x$ also governs the values of $\lambda_i$ with \eqref{rec3}.
        \end{proof}
        
        We prove that Algorithm~\ref{algo1} solves the fixed-time problem by first reformulating the optimisation program to reveal its structure. Then we derive the steps of the algorithm in a similar fashion to Cantoni {\it et al.}~\cite{cantoni2020structured}.
        
        \begin{proof}[Proof of Proposition~\ref{prop::alg_complexity}]
        Similar to Richter {\it et al.}~\cite{Richter2016}, we reformulate \eqref{eq::form2} as the following.
        \begin{subequations}\label{eq::reform1}
            \begin{align}
                \min_{f} \quad &f^TV^{-T}HV^{-1}f, \\
                \text{s.t. } \quad & Ef =0, \\
                \quad & P f = b.
            \end{align}
        \end{subequations}
        
        We make the change of variable $f = \overline{f} + Z g$ such that $P \overline{f}=b$ and  $PZ=0$. By construction $P f=P(\overline{f} + Z g)= b$, and the program resulting from the substitution is
        \begin{subequations} \label{eq::form4}
            \begin{align}
            \min_{g} \quad &2g^TZ^TV^{-T}HV^{-1}\overline{f} + g^TZ^TV^{-T}HV^{-1}Z g, \\
            \text{s.t. } \quad & E Z g =0. \label{eq::form4::constr}
        \end{align}
        \end{subequations}
        
        % Cast as \eqref{eq::form4}, the matrices in the program present a lot of structure. 
        Note that, the Hessian $Z^TV^{-T}HV^{-1}Z$ is block diagonal and the decision variables are coupled by the sparse matrix $E Z$. Furthermore, the coupling is only between the derivatives of adjacent segments. The recursions introduced in Lemma~\ref{lemm::recursions} lead to the steps of Algorithm~\ref{algo1} as described below. 
        
        The KKT conditions for \eqref{eq::form4} yield \eqref{eq::KKT}. Lemma~\ref{lemm::recursions} solves \eqref{eq::form4} with linear computational complexity in $l$. Through a change of variables, the recursions can be used to calculate the solution to \eqref{eq::form1}. There are $2l$ matrix calculations required in computing \eqref{con1and2}, while $2l$ systems of equations need to be solved in \eqref{rec1} and \eqref{rec2}. All the matrices involved are square with dimensions smaller than or equal to $k$. Hence, the computational complexity is $O(k^3l)$.
        % \begin{align*}
        %     \begin{bmatrix}
        %     Z^TV^{-T}HV^{-1}Z & Z^TE^T \\
        %     EZ & 0
        %     \end{bmatrix}\begin{bmatrix}
        %     f \\ \lambda
        %     \end{bmatrix} &= \begin{bmatrix}
        %     Z^TV^{-T}HV^{-1}\overline{f} \\ 0
        %     \end{bmatrix},
        % \end{align*}
        % where the Lagrangian multipliers are $\lambda = [\lambda_1^T, \dots, \lambda_{k-1}^T]^T$. % and $\lambda_i\in \mathbb{R}^s$. % We partition $g$ the same as $f$, $g=[g_{1,0}^T,g_{1,T}^T, \ldots,g_{k,0}^T,g_{k,T}^T]^T$.
        \end{proof}
    \section{Proof of Proposition~\ref{prop::scaling}}\label{app:scaling}
    The cost function \eqref{eq::form1::cost} can be written as the summation 
\begin{align*}
    J(t) &= \min_{\sigma(\tau),f}\sum_{i=1}^l \int_{\tau_{i-1}}^{\tau_i}\big(\sigma^{(k-1)}(\tau)\big)^2 d\tau,
    \intertext{Under the change of variables $\rho=2\tau/(\tau_i - \tau_{i-1}) - (\tau_i+\tau_{i-1})/(\tau_i - \tau_{i-1})$ for each integral in the summation}
    J(t) &= \min_{\widetilde{\sigma}(\tau),f}\sum_{i=1}^l \Big(\frac{2}{\tau_i-\tau_{i-1}}\Big)^{2k-1} \int_{-1}^{1}\big(\widetilde{\sigma}^{(k-1)}(\rho)\big)^2d\rho, \\
    &= \min_{\widetilde{a}_i,f}\sum_{i=1}^l \Big(\frac{2}{\tau_i-\tau_{i-1}}\Big)^{2k-1} \widetilde{a}_i^T H_i(1,-1) \widetilde{a}_i, \\
    &= \min_{\widetilde{a},f}\widetilde{a}^T \widetilde{H}(t) \widetilde{a}.
\end{align*}
% \begin{align*}
%     J(t) &= \min_{\widetilde{\sigma}(\tau),f}\sum_{i=1}^l \Big(\frac{2}{\tau_i-\tau_{i-1}}\Big)^{2k-1} \int_{-1}^{1}\widetilde{\sigma}_i(\rho)d\rho,
% \intertext{where $\widetilde{\sigma}_i(\rho)=\sigma_i((\tau_i-\tau_{i-1})\rho/2 + (\tau_i+\tau_{i-1})/2)$. Let $\widetilde{a}_i$ be the coefficients of the shifted and scaled polynomials $\widetilde{\sigma}_i(\rho)$ for $i=1,\ldots,l$ and $\widetilde{a}=\mathrm{vec}\{\widetilde{a}_1,\ldots, \widetilde{a}_l\}$. Thereby}
%     J(t) &= \min_{\widetilde{a}_i,f}\sum_{i=1}^l \Big(\frac{2}{\tau_i-\tau_{i-1}}\Big)^{2k-1} \widetilde{a}_i^T H_i(1,-1) \widetilde{a}_i, \\
%     &= \min_{\widetilde{a},f}\widetilde{a}^T \widetilde{H}(t) \widetilde{a}.
% \end{align*}
Therefore, subject to the constraints \eqref{eq::form2::vand}--\eqref{eq::form2::phi}, if $a^\star$ and $f^\star$ minimise \eqref{eq::form2::cost} then $\widetilde{a}^\star$ and $f^\star$ minimise \eqref{eq::form_nondim::cost}.
% Therefore, the cost functions \eqref{eq::form2::cost} and \eqref{eq::form_nondim::cost} are equal. 
We now show that the constraints \eqref{eq::form_nondim::vand}--\eqref{eq::form_nondim::const2} are equivalent to \eqref{eq::form2::vand}--\eqref{eq::form2::phi} to complete the proof. Under the change of variables, for $i=1,\ldots,l$ and $q=1,\ldots,k$,
\begin{align*}
    \sigma^{(q-1)}_i(\tau_{i-1})&=\Big(\frac{2}{\tau_i-\tau_{i-1}}\Big)^{q-1}\widetilde{\sigma}^{(q-1)}_i(-1), \\
    \sigma^{(q-1)}_i(\tau_{i})&=\Big(\frac{2}{\tau_i-\tau_{i-1}}\Big)^{q-1}\widetilde{\sigma}^{(q-1)}_i(1).
\end{align*}
This can be written more compactly as $V(t)a=G(t)\widetilde{V}\widetilde{a}$ and thereby \eqref{eq::form_nondim::vand} is equivalent to \eqref{eq::form2::vand}. The remaining constraints are in common.\hfill $\square$
    
    \section{Proof of Proposition~\ref{prop::variable}} \label{app::variable}
    We first prove the necessary condition. To obtain a contradiction, assume $t^\star=Rd^\star$ is a local minimiser of \eqref{eq::form3} but $d^\star$ is not a local minimiser of \eqref{eq::form6}. There exists some $d^\dagger$ such that $\lVert d^\dagger - d^\star \rVert < \varepsilon$ and $J(Rd^\dagger)<J(Rd^\star)$. Hence, there exists a $t^\dagger=Rd^\dagger$ such that $\lVert t^\dagger - t^\star \rVert < \lVert R \rVert \varepsilon$ and $J(t^\dagger)<J(t^\star)$, which contradicts $t^\star$ being a local minimiser of \eqref{eq::form3}.

To prove the sufficient condition, we first state a useful fact. Consider the times $t_1,~t_2\in\mathbb{R}^{l+1}$ and their differences $d_1,~d_2\in\mathbb{R}^{l}$, respectively, such that $d_1=d_2$. The formulation \eqref{eq::form_nondim} depends only on the difference between times $\tau_i-\tau_{i-1}$, $i=1,\ldots,l$, therefore $J(t_1)=J(t_2)$. We now prove the sufficient condition by contradiction. Assume that $d^\star$ is a local minimiser of \eqref{eq::form6} but $Rd^\star$ is not a local minimiser of \eqref{eq::form3}. Therefore, a $t^\dagger$ exists such that $\lVert t^\dagger - R d^\star \rVert < \varepsilon$ and $J(t^\dagger)<J(Rd^\star)$. However, as the formulation depends only on the difference between times, there must also exist some $d^\dagger$ such that $J(t^\dagger)=J(Rd^\dagger)$. Hence, there is a $d^\dagger$ such that $\lVert d^\dagger - d^\star \rVert < \varepsilon / \lVert R \rVert $ and $J(Rd^\dagger)<J(Rd^\star)$, which contradicts the assertion that $d^\star$ is a local minimiser of \eqref{eq::form6}. \hfill $\square$
    \section{Proof of Proposition~\ref{prop::gradient}} \label{app::gradient}
    We start by performing a similar reformulation as in the proof of Proposition~\ref{prop::alg_complexity}. We note that $G(t)$ depends only on the difference between times $\delta_i=(\tau_i-\tau_{i-1})/2$, so without loss of generality we reparametrise it as $G(d):\mathbb{R}^{l}\rightarrow\mathbb{R}^{2kl\times 2kl}$. This allows us to recast \eqref{eq::form6} as
    \begin{subequations}\label{eq::reform2}
        \begin{align}
            J(Rd)=\min_{f} \quad &f^T(G(d)\widetilde{V})^{-T}\widetilde{H}(G(d)\widetilde{V})^{-1}f, \\
            \text{s.t. } \quad & Ef = 0, \\
            & Pf = b.
        \end{align}
    \end{subequations}
    We make the change of variable $f = \overline{f} + Zg$ such that $P\overline{f}=b$ and  $PZ=0$. Further, we let $g=Yh$ such that $EZY=0$ and $Y\in\mathbb{R}^{(kl-m/2)\times (2kl-m)}$. Substituting the two change of variables into \eqref{eq::reform2} results in the unconstrained optimisation program
    \begin{align}\label{eq::unconstrained}
        J(Rd)=\min_{h} \quad h^T Q(d) h + h^Tq(d) + \overline{q}(d),
    \end{align}
    where
        \begin{align*}
            Q(d) &= Y^TZ^T (G(d)\widetilde{V})^{-T}\widetilde{H}(G(d)\widetilde{V})^{-1} ZY, \\
            q(d) &= 2 Y^TZ^T (G(d)\widetilde{V})^{-T}\widetilde{H}(G(d)\widetilde{V})^{-1} \overline{f}, \\
            \overline{q}(d) &= \overline{f}^T(G(d)\widetilde{V})^{-T}\widetilde{H}(G(d)\widetilde{V})^{-1}\overline{f}.
        \end{align*}

    Let $h^\star(d)$ be the argument that minimizes \eqref{eq::unconstrained} such that 
    \begin{align*}
        J(Rd)&=h^\star(d)^T Q(d) h^\star(d) + h^\star(d)^T q(d) + \overline{q}(d).
    \end{align*}
    %  note that no new variables will be introduced and so we will stop explicitly stating time dependence, i.e., we will write $J(t),~Q(t),~q(t),~\overline{q}(t),~h^\star(t)$ as $J,~Q,~q,~\overline{q},~h^\star$.
    
    % The partial derivative of $J(t)$ with respect to time $t$ is then
    % \begin{align*}
    %     \nabla_t J(t) &=\big(\nabla_t h^\star(t)\big) ^TQ(t) h^\star(t) + h^\star(t)^T \big(\nabla_t Q(t)\big) h^\star(t)  \\ 
    %     & + \frac{1}{2}h^\star(t)^T Q(t) \nabla_t h^\star(t)
    %     + (\nabla_t h^\star(t))^Tq(t) +h^\star(t)^T \nabla_t q(t) + \nabla_t \overline{q}(t). 
    % \end{align*}
    A necessary condition for the optimality of the fixed-time minimum-snap trajectory is \begin{align}\label{eq::opt_cond_1}
        2Q(d)h^*(d)+q(d)&=0.
    \end{align}
    Substituting \eqref{eq::opt_cond_1} into $\nabla_d J(Rd)$ yields
    \begin{align} \label{eq::partial_J_opt}
        \nabla_d J(Rd) &= h^\star(d)^T \big(\nabla_d Q(d)\big) h^\star(d) \\
        & \quad + h^\star(d)^T \nabla_d q(d) + \nabla_d \overline{q}(d). \nonumber
    \end{align}
    We will now identify the nonzero components of $\nabla_d Q(d)$, $\nabla_d q(d)$ and $\nabla_d \overline{q}(d)$ and use them to form compact expressions for $\nabla_d J(Rd)$. 
    For $i=j$, the partial derivatives of the block-diagonal components of $G(d)$ and $\widetilde{H}$
    % \declan[inline]{Put two $i=j$ cases together. And change $G_0$}
    \begin{align*}
        \frac{\partial G_j}{\partial \delta_i} &= \mathrm{diag}\{0, -\delta_j^{-2}, \ldots,-(s-1)\delta_j^{-s}, \ldots \\
        & \quad 0, -\delta_j^{-2}, \ldots,-(s-1)\delta_j^{-s}\},
    \end{align*}
    and
    \begin{align}
        \frac{\partial}{\partial \delta_j}\big(\frac{1}{\delta_i^{2k-1}}H_i(-1,1)\big) = \frac{1-2k}{\delta_i^{2k}}H_i(-1,1). \label{eq::prop2_p2}
    \end{align}
    Otherwise, the partial derivatives of the block-diagonal components of $G(d)$ and $\widetilde{H}$ are zero for $i\neq j$.
    
    To ease notation in the final expression we introduce
    \begin{align}\label{eq::prop2_p1}
        F(\delta_i)&=\frac{\partial G_i(\tau_{i-1},\tau_i)}{\partial \delta_j}G_i\tau_{i-1},\tau_i)^{-1}, \\
        &= \mathrm{diag}\{0, -\delta_j^{-1}, \ldots,-(s-1)\delta_j^{-1}, \ldots \nonumber \\
        & \quad \quad 0, -\delta_j^{-1}, \ldots,-(s-1)\delta_j^{-1}\}. \nonumber
    \end{align}
    
    The last calculation required is
    \begin{align} \label{eq::prop2_p3}
        \nabla_d (G(d) \widetilde{V})&= - \widetilde{V}^{-1} (\nabla_d G(d)) G^{-1}(t).
    \end{align}
    
    Substituting \eqref{eq::prop2_p1}, \eqref{eq::prop2_p2} and \eqref{eq::prop2_p3} into the partial derivatives $\nabla_d Q(d),~\nabla_d q(d)$ and $\nabla_d \overline{q}(d)$ then evaluating \eqref{eq::partial_J_opt} yields the compact expressions \eqref{eq::partialJ} as required.\hfill $\square$
\ifCLASSOPTIONcaptionsoff
  \newpage
\fi



\bibliographystyle{IEEEtran}
\bibliography{IEEEabrv, mybibfile}

% \end{IEEEbiography}

\end{document}